\documentclass[a4paper,10pt]{book}
\usepackage{array}
\usepackage{amssymb}
\usepackage{bm}
\usepackage{enumitem}

\usepackage{natbib}
\usepackage{amsfonts}
\usepackage{graphicx}
\usepackage[usenames,dvipsnames]{xcolor}
\usepackage{amsmath}
\usepackage{amsthm}
\usepackage{hyperref}

\theoremstyle{plain}
\newtheorem{theorem}{Theorem}
\newtheorem{proposition}{Proposition}
\newtheorem{corollary}{Corollary}
\newtheorem{lemma}{Lemma}

\theoremstyle{definition}

\newtheorem{example}{Example}
\newtheorem{excont}{Example}

\newcommand{\eff}{\operatorname{eff}}
\newcommand{\Pri}{\mathcal{P}}
\newcommand{\E}{E}
\renewcommand{\vec}{\bm}
\newcommand{\T}{\textsc{T}}
\newcommand{\expit}{\operatorname{expit}}
\newcommand{\argmax}{\operatorname{arg \,max}}

\newcommand{\fx}[1]{\vec{f}(\vec{x}_{#1})}

\newcommand{\bbeta}{\bm{\beta}}
\newcommand{\Q}{\mathcal{Q}}
\newcommand{\cS}{\mathcal{S}}
\newcommand{\xiQL}{\xi^\ast_{\Q,L}}
\newcommand{\xiQU}{\xi^\ast_{\Q,U}}
\newcommand{\cP}{\mathcal{P}}
\newcommand{\given}{\,|\,}

\renewcommand{\bm}{} 

\renewcommand{\chapter}[1]{  } 

\newcommand{\mychapter}[1]{
    \refstepcounter{chapter}
    \setcounter{section}{0}
    \setcounter{equation}{0} 
    \par
    \vspace{1em}
    \noindent {\bf \thechapter. #1}
}

\newcommand{\mysection}[1]{
   \refstepcounter{section}
   
   \noindent {\bf \thesection. #1}
}

\newcommand{\mysubsection}[1]{
   \refstepcounter{subsection}
   
   \noindent {\bf \thesubsection. #1}
}

\begin{document}
\renewcommand{\baselinestretch}{1.2}
\markright{
}

\markboth{\hfill{\footnotesize\rm TIMOTHY W.~WAITE 
}\hfill}
{\hfill {\footnotesize\rm  SINGULAR PRIORS IN BAYESIAN DESIGN} \hfill}
\renewcommand{\thefootnote}{}
$\ $\par
\fontsize{10.95}{14pt plus.8pt minus .6pt}\selectfont
\vspace{0.8pc}
\centerline{\large\bf  SINGULAR PRIOR DISTRIBUTIONS  {AND}  }
\vspace{2pt}
\centerline{\large\bf  { ILL-CONDITIONING} IN BAYESIAN \emph{D}-OPTIMAL }
\vspace{2pt}
\centerline{\large\bf  DESIGN FOR {SEVERAL} NONLINEAR MODELS}
\vspace{.4cm}
\centerline{Timothy W. Waite}
\vspace{.4cm}
\centerline{\it  University of Manchester}
\vspace{.55cm}
\fontsize{9}{11.5pt plus.8pt minus .6pt}\selectfont

\begin{quotation}
\noindent {\it Abstract:}
For Bayesian $D$-optimal design, we define a \emph{singular  prior distribution} for the model parameters as a prior distribution such that the determinant of the Fisher information matrix has a prior geometric mean of zero for all designs. 
For such a prior distribution, the Bayesian $D$-optimality criterion fails to select a design.  For the exponential decay model, we characterize singularity of the prior distribution in terms of the expectations of a few elementary transformations of the parameter. For a compartmental  model and {several multi-parameter generalized linear models,} we establish sufficient conditions for singularity of a prior distribution. For {the generalized linear models} we also obtain sufficient conditions for non-singularity.
In the existing literature,  weakly informative prior distributions are commonly recommended as a default choice for inference in logistic regression. Here it is shown that some of the recommended prior distributions are singular, and hence should not be used for Bayesian $D$-optimal design.   Additionally, methods are developed to derive and assess Bayesian $D$-efficient designs when numerical evaluation of the objective function fails due to ill-conditioning, {as often occurs for heavy-tailed prior distributions. These numerical methods are illustrated for logistic regression.} \par

\vspace{9pt}
\par
\end{quotation}\par

\fontsize{10.95}{14pt plus.8pt minus .6pt}\selectfont
\mychapter{Introduction}

In recent years, much effort has been devoted to the development of $D$-optimal design methods for nonlinear problems; for example, nonlinear models (e.g.~\citet*{yang-stuf,yang2012identifying,yang2010garza}), generalized linear models (\citet*{khuri2006design,woods06, yang11,yang2015d}), and linear models with mixed effects (\citet*{jones2009d}).  In each of these areas, the choice of a $D$-optimal design depends on the unknown vector of model parameters, $\bm{\theta} \in \Theta \subseteq \mathbb{R}^p$. 

One approach to choosing a design is to make a `best guess' of the parameter values, and calculate a corresponding locally $D$-optimal design \citep{chernoff1953locally}, i.e.~$\xi^\ast_\bm{\theta} \in \argmax_{\xi \in \Xi} |M(\xi; \theta)|$, where $M(\xi;\theta)$ is the Fisher information matrix for design $\xi \in \Xi$, and $\Xi$ is the set of all competing designs. However, the performance of a locally optimal design  may be highly sensitive to misspecification of the value of $\theta$. 
Then a Bayesian approach is often used to derive designs that are efficient for a variety of plausible values for $\theta$. This approach requires the adoption of a prior distribution, $\Pri$, on the parameters, and maximization of the value of an objective function that quantifies the expected information contained in the experiment. Throughout, we assume that $\cP$ is a probability measure on the measure space $(\Theta, \Sigma)$, with $\Sigma$ the Borel $\sigma$-algebra over $\Theta$.
A widely used objective function is 
\begin{equation}
\phi(\xi; \Pri) = 
\int_\Theta \log 
	|  M(\xi ; {\bm{\theta}})| \,
	d\Pri(\bm{\theta})
\label{eq:bayes-d}  \,,
\end{equation}
see, for example, \citet*{chaloner-larntz} and \citet*{gotwalt}. We adopt the measure-theoretic formulation of integration, under which the notation $\int_\Theta g (\theta)\,d\cP(\theta)=\infty$ is standard when $g:\Theta \to \mathbb{R}$ is a non-negative $\Sigma$-measurable function (\citet{capinski2004measure}, pp.~77-8).
When $g:\Theta \to \mathbb{R}$ is a general $\Sigma$-measurable function, it is said that $\int_\Theta g(\theta)\, d\cP(\theta)=-\infty$ if and only if $\int_\Theta g^{+}(\theta)\, d\cP(\theta) < \infty$ and $\int_\Theta g^{-}(\theta)\, d\cP(\theta) = \infty$, where $g^{+}(\theta)= \max\{0,g(\theta)\} $ and $g^{-}(\theta)= \max\{0, -g(\theta) \}$.

A design that maximizes \eqref{eq:bayes-d} is said to be \emph{(pseudo-)Bayesian $D$-optimal}, and may be used whether or not a Bayesian analysis will be performed (e.g.~\citet*{woods06}). Maximization of \eqref{eq:bayes-d} is equivalent to maximization of an asymptotic approximation to the Shannon information gain from prior to posterior (\citet*{chaloner1995bayesian}).

In nonlinear problems, a \emph{singular parameter vector} is a $\theta$ such that $M(\xi;\bm{\theta})$ has determinant zero for \emph{any} design $\xi \in \Xi$. 
For such $\bm{\theta}$, it is difficult to estimate the parameters no matter which design is used, often because of a lack of model  identifiability (see Section \ref{sec:logistic}). 
In this situation,  the local $D$-optimality criterion fails to select a design.  The analogue of a singular parameter vector for Bayesian $D$-optimality is defined through: 

\begin{enumerate}[label={(\alph*)}]
\itemsep 0pt \parskip 0pt \parsep 0pt
\item Given  $\xi \in \Xi$ and a prior distribution, $\cP$, we say that $\xi$  is a \emph{Bayesian  singular design  with respect to $\cP$} if $\phi(\xi; \cP) = -\infty$.

\item {Given a prior distribution, $\cP$, we say that $\cP$ is a \emph{singular prior distribution} if \emph{all} $\xi \in \Xi$ are Bayesian singular with respect to $\cP$, or equivalently if the geometric mean of $|M(\xi; \theta)|$ under $\cP$ is zero for all $\xi\in\Xi$. Above, the geometric mean of a non-negative random variable $X$ is defined as $\E^\mathcal{G}(X)= \exp[ \E \{ \log(X)\}]$, with $\E^\mathcal{G}(X) = 0$ if $\E\log(X)=-\infty$ (\citet{feng2015generalized}).}
\end{enumerate}

{For a singular prior distribution $\cP$, Bayesian $D$-optimality cannot be used to select a design, since all designs have the same objective function value $\phi(\xi; \cP)=-\infty$.}
In many models, such as the exponential decay model and logistic regression, it is straightforward to detect singular parameter vectors, $\theta$, by inspection of the information matrix. However, as shown below, it is more difficult to detect whether $\Pri$ is a singular prior distribution, except in the case of point priors.

A different, but related, problem is the presence of  ill-conditioned information matrices in a quadrature approximation  to \eqref{eq:bayes-d}. {For several models, this is likely to occur for a heavy-tailed prior distribution $\cP$, even if $\cP$ is theoretically non-singular.} 
Such ill-conditioning causes failure of numerical selection of Bayesian $D$-optimal designs.
  
In this paper, we clarify and extend the set of prior distributions for which Bayesian $D$-optimal design is feasible for three important classes of models. In Sections \ref{sec:exp}, \ref{sec:compart}, and \ref{sec:logistic}, respectively, we give  examples of singular prior distributions for the one-factor exponential decay model, a three-parameter compartmental model, and {several} multi-factor {generalized linear models.} In Section \ref{sec:logistic} the default weakly informative prior proposed for logistic regression by \citet*{gelman2008weakly} is shown to be singular. For the exponential and {generalized linear} models,  sufficient conditions for a prior distribution to be non-singular are established. These conditions are easily checked to determine if the Bayesian $D$-optimality criterion can be used to select designs under $\cP$. In Section \ref{sec:illcond}, novel methods are developed that enable the selection of highly Bayesian $D$-efficient designs for logistic regression when the  quadrature approximation to \eqref{eq:bayes-d} is ill-conditioned, {thereby facilitating design for heavy-tailed prior distributions.} Finally, in Section \ref{sec:discussion} we discuss alternative approaches to finding efficient designs when $\cP$ is a singular prior distribution.

\mychapter{Singularity of prior distributions for {some} standard models}

\mysection{Exponential decay model}
\label{sec:exp}

We derive necessary and sufficient conditions for a prior distribution to be singular for the exponential decay model  which is used, for example, to model the concentration of a chemical compound over time. {This model is commonly used as a simple illustrative example of a nonlinear model in the optimal design of experiments literature, e.g. \citet{dette1997bayesian,atkinson2003horwitz}; the results here help to develop our intuition.}
Here, two parameterizations are considered: by rate, $\beta>0$, and by `lifetime', $\theta = 1/\beta>0$. For the former, the model for the response, $y$, in terms of explanatory variable, $x>0$, is 
\begin{alignat*}{3}
y_i =
e^{-\beta x_i}  + \epsilon_i 
	\,,& \hspace{0.5cm} 
		\epsilon_i  \sim N(0,\sigma^2)\,, 
\end{alignat*}
where $i=1,\ldots, n$, $x_i\geq0$, and $\sigma >0$.

Assume that $\Xi = \mathcal{X}^n$, where $\mathcal{X}=[0,\infty)$. Then design $\xi = ( x_1,\ldots, x_n ) \in \Xi$ has information matrix
\begin{align*}
M_\beta (\xi ; \beta) = 
\sum_{i=1}^{n} 
	x_i^2  e^{-2\beta x_i } \,.
\end{align*}

Suppose that at least one $x_i >0$ and let $S_{xx}= \sum_{i=1}^{n}{x_i^2}$. Then
\begin{equation}
-2\beta \max_{i=1,\ldots,n}\{ x_i\} 
\leq 
\log|M_\beta (\xi ; \beta) | -  \log S_{xx}\leq 
-2 \beta \min_{i \,:\, x_i > 0}\{x_i  \} \,. 
\label{eq:exp-inequal}
\end{equation}
By taking expectations, the following result is obtained.
\begin{proposition}
Suppose that at least one $x_i>0$. Then, for the $\beta$-parameterization, $\phi(\xi;\Pri)>-\infty$ if and only if $\E_\Pri( \beta ) < \infty$. 
 \end{proposition}
 
{Here the prior, $\cP$, is non-singular provided the rate parameter has finite expectation, but $\cP$ can be singular if the distribution of $\beta$ is heavy-tailed with infinite mean, e.g.~if $\beta$ is half-Cauchy (cf.~\citet{polson2012half}).}
  
 For the $\theta$-parameterization,  a change-of-variable argument shows that 
 \begin{equation}
 \log |M_\theta(\xi; \theta)| = \log |M_\beta(\xi; \beta)| - 4\log\theta \,.
 \label{eq:exp-reparam}
 \end{equation} This enables derivation of the following result; for proof see the supplementary material.
\begin{proposition}
For the $\theta$-parameterization, the prior distribution $\cP$ is singular if and only if either 
$\E_{\cP} ( 1/\theta ) = \infty$ or $\E_{\cP} ( \log\theta )=\infty$. 
\label{prop:exponential}
\end{proposition}
In the context of designs that maximize $\phi(\xi;\cP)$ for nonlinear models, \citet{chaloner1995bayesian} refer to potential `technical problems using prior distributions with unbounded support where [\ldots] $M(\xi;\theta)$ may be arbitrarily close to being singular'.  Corollary \ref{cor:exp-unif} below shows that, even with bounded support, seemingly innocuous prior distributions can cause Bayesian $D$-optimality to fail as a design selection criterion.
\begin{corollary}
\label{cor:exp-unif}
For the $\theta$-parameterization, the prior distribution $\Pri = U(0,a)$, $a>0$, is  singular.
\end{corollary}
{Note that under the prior for $\theta$ in Corollary \ref{cor:exp-unif}, the corresponding implied prior distribution for $\beta$ has a proper density, $p(\beta) = 1/(a\beta^2)$ for $\beta \geq 1/a$. However, this implied distribution for $\beta$ has unbounded support and is heavy tailed, such that $\E(\beta) = \infty$. In other words, the implied a priori expectation is that the decay is very rapid.}
\\[0em]
\mysection{Compartmental model}
\label{sec:compart}

In this section, we derive sufficient conditions for a prior distribution to be  singular for the following three-parameter compartmental model:\begin{equation}
y_i =  \theta_3 \{ e^{-\theta_1 x_i} - e^{-\theta_2 x_i} \}  + \epsilon_i\,,
\qquad \epsilon_i \sim N(0,\sigma^2) \,,
\label{eq:compart-standard}
\end{equation}
where $x_i \geq 0$, $i=1,\ldots,n$, $\theta_2 > \theta_1>0$, $\theta_3 >0$ and $\sigma>0$. Here, $\Xi = [0,\infty)^n$. 
As with the exponential model, often the response $y_i$ is a concentration of a compound in a system, and the $x_i$ are the observation times.
{For example, \citet{atk93} consider a theophylline kinetics experiment on horses, finding optimal sampling times for model \eqref{eq:compart-standard} under several different (pseudo-)Bayesian criteria.}

The information matrix for the $i$th time point is
$$
M(x_i ; \theta) = \begin{pmatrix}
	x^2_i \theta^2_3 e^{-2\theta_1 x_i} 
		&	-x^2_i \theta^2_3 e^{-(\theta_1+\theta_2)x_i}
		&	-f_i x_i  e^{-\theta_1 x_i}	\\
	-x^2_i \theta^2_3 e^{-(\theta_1+\theta_2)x_i }   
		&      x^2_i \theta^2_3 e^{-2\theta_2 x_i}
		& 	f_i x_i  e^{-\theta_2 x_i}	\\
	-f_i x_i  e^{-\theta_1 x_i}  
		&	f_i x_i e^{-\theta_2 x_i} 
		&	f^2_i / \theta^2_3
\end{pmatrix} \,,
$$
where $f_i = \theta_3 \{ e^{-\theta_1 x_{i}} - e^{-\theta_2 x_{i}} \}$. We have $|M(\xi;\theta)|=0$ when { (i) }$\theta_1=\theta_2$ or  { (ii) }$\theta_3=0$, and  $|M(\xi; \theta)|\to 0$ when { (iii) }$\theta_1 \to \infty$. {Physically,  conditions (i) and (iii) correspond to situations where the flow rates in and out of the compartment are either exactly balanced, or both very rapid. Each of the potentially very different parameter scenarios in (i)--(iii) results in a similar response profile, in which the concentration is close to zero throughout the duration of the experiment. Thus, if such a profile is observed, it is difficult to ascertain which values of the parameters generated the data. }

From the above, it is clear that for $\cP$ to be a non-singular prior distribution { its probability density must not be too highly concentrated near regions where $\theta_2=\theta_1$ or $\theta_3=0$, nor can the prior for $\theta_1$ be too heavy-tailed.}  This is formalized by Proposition \ref{thm:compart-model}, for which the following two lemmas are required; proofs are given in the supplementary material. Let $\delta = \theta_2 - \theta_1 > 0$.

\begin{lemma}
\label{lemma:compart-bounds}
We have the following bounds on $\log|M(\xi;\theta)|$, 
\begin{equation*}
- 6\theta_1 x_{\max}
\leq
	\log |M(\xi;\theta)| - 4 \log \theta_3
	- \log |\tilde{M}_{\delta,1}|
\leq
-6\theta_1 x_{\min} \,,
\end{equation*}
where  
$\tilde{M}_{\delta,1}$ is 
$\tilde{M}_{\delta,\theta_3} 
$ evaluated at $\theta_3=1$, and
\begin{align*}
&\tilde{M}_{\delta,\theta_3} = \sum_{i=1}^{n} \tilde{M}_{\delta,\theta_3}^{(i)} \,,
\quad
x_{\min}  = 
	\min_{i \,:\,x_i>0}\{ x_i\} \,, \quad 
x_{\max} =
	\max_{i=1,\ldots,n}\{ x_i   \}  \,, 
\\
\tilde{M}_{\delta, \theta_3}^{(i)} 
&=
\begin{pmatrix}
 	x^2_i \theta_3^2
		&	-x^2_i \theta^2_3 e^{-\delta x_i}
		&	-x_i \theta_3 (1-e^{-\delta x_i}) 	\\
	- x^2_i \theta_3^2 e^{-\delta x_i}
		& 	x^2_i \theta^2_3 e^{-2\delta x_i}
		&	x_i \theta_3 e^{-\delta x_i} (1-e^{-\delta x_i}) \\
	- x_i \theta_3 (1-e^{-\delta x_i}) 
		&	x_i \theta_3 e^{-\delta x_i} (1-e^{-\delta x_i})
		& 	(1-e^{-\delta x_i})^2
\end{pmatrix} \,.
\end{align*}

\end{lemma}

\begin{lemma}
\label{lemma:compart-mtilde}
If   $\int_{\delta < 1} \log \delta \,d\cP(\theta)= -\infty$,  
 then $\E_\cP ( \log|\tilde{M}_{\delta,1}| ) = -\infty$.
\end{lemma}
\begin{proposition}
\label{thm:compart-model}
Suppose $\int_{\theta_3>1} \log \theta_3 \,d\cP (\theta)< \infty$. For model \eqref{eq:compart-standard}, the prior parameter distribution $\Pri$ is singular if  $\E_{\cP}( \theta_1 ) = \infty$, $\int_{\theta_3<1} \log \theta_3 \,d\cP (\theta)=-\infty $, or $\int_{\delta<1} \log \delta\, d\cP (\theta) =-\infty$.
\end{proposition}
 Heavy-tailed priors such as the half-Cauchy are increasingly recommended as weakly informative priors in various models (\citet{gelman2008weakly}; \citet{polson2012half}). For model \eqref{eq:compart-standard}, $\cP$ is singular if $\theta_1$ is half-Cauchy distributed, although for physiological compartmental models more specific prior information is often used (\citet{gelman1996physiological}). 
\\[0em]
\mysection{{Generalized linear models}}
\label{sec:logistic}

Suppose there are $n$ design points  $\vec{x}_i = (x_{i1},\ldots, x_{iq})^\T \in \mathcal{X}$,  {with responses $y_i$, $i=1,\ldots,n$.} 
We assume a generalized linear model (GLM; \citet*{mccullagh-nelder}), {thus $y_i$ has an exponential family distribution with mean $\mu_i = \mu(x_i; \beta)$ and variance $\gamma v(\mu_i)$, where $\mu$ satisfies
\begin{equation}
h[ \mu(x; \beta) ]  =  \eta(x; \beta)= \vec{f}^\T(\vec{x})\bm{\beta}\,,
\label{eq:linear-predictor}
\end{equation}
with $h$ the link function, $\gamma$ a dispersion parameter, $v$ the variance function, and $\eta_i = \eta(x_i; \beta)$ the linear predictor. For binomial and Poisson responses, $\gamma=1$ with variance function $v(\mu)=\mu(1-\mu)$ and $v(\mu)=\mu$, respectively.} Above, $\vec{f}(\vec{x}) = (f_0(\vec{x}),\ldots, f_{p-1}(\vec{x}))^\T$ contains regression functions \mbox{$f_j: \mathcal{X} \to \mathbb{R}$}, $j=0,\ldots,p-1$, and $\bm{\beta}= (\beta_0,\beta_1,\ldots ,\beta_{p-1})^\T \in \Theta$ is a vector of $p$  regression parameters. We let $\mathcal{X}=[-1,1]^q$ and $\Xi = \mathcal{X}^n$.

For design $\xi = ( \vec{x}_1 ,\ldots, \vec{x}_n )$  and model \eqref{eq:linear-predictor} 
\begin{align}
M(\xi; \bbeta) &= \sum_{i=1}^{n} w_i \,\vec{f}(\vec{x}_i) \vec{f}^\T(\vec{x}_i) 
\notag \\ 
w(\eta) &= 
{ \frac{1}{\gamma v(\mu)}\left( \frac{\partial \mu}{ \partial \eta }\right)^2  \,,}
\label{eq:glm-weight}
\end{align}
{with $w_i = w(\eta_i)$, $i=1,\ldots,n$ (e.g.~\citet{khuri2006design}, \citet{atkinson2015designs}, \citet{yang2015d}).}

The following  {lemmas are important first steps towards the derivation of} results on singular prior distributions. {Lemma \ref{lemma:info-bound} also}  facilitates the development of numerical methods to overcome ill-conditioning in Section \ref{sec:illcond}. {The proofs are} straightforward; the details are omitted. Let $F$ be the model matrix with rows $\vec{f}^\T(\vec{x}_i)$, noting that 
$\sum_{i=1}^{n} 
	\vec{f}(\vec{x}_i) 
	\vec{f}^{\T}(\vec{x}_i) =F^\T F$ 
is the information matrix of $\xi$ under a linear model with regressors specified by $\vec{f}$. The inequality below is with respect to the Loewner partial ordering on real symmetric matrices, in which $M_1 \preceq M_2$ if and only if $ M_2 - M_1$ is non-negative definite {(for example, \citet[p.11]{pukelsheim1993optimal}).} 

\begin{lemma}
For {a generalized linear model,} the information matrix satisfies
$$
\min_{i=1,\ldots,n} \{ w_i \} 
F^\T F
\preceq 
M(\xi;\bbeta) 
\preceq 
\max_{i=1,\ldots,n} \{ w_i \} 
F^\T F \,.
$$
{Thus, since the log-determinant respects the Loewner ordering,
\begin{align*}
p\log \min_{i} \{ w_i  \} + \log | F^\T F | 
 \leq
  \log | M(\xi;\bbeta) | 
  \leq
  p\log \max_i \{ w_i \} + \log |F^\T F| \,.
\end{align*}}
\label{lemma:info-bound}
\end{lemma}
{\begin{lemma}
Suppose that $\xi$ is non-singular for the linear model with regressors given by $\vec{f}$, that is 
$
\left| F^\T F
 \right|>0 
$. Then we have the following:
\begin{enumerate}[label=\emph{(\roman*)}]
\item 
If $E_\Pri \{ \log \min_i w_i \} > -\infty$, then $\phi(\xi; \Pri) > -\infty$, i.e.~$\xi$ is Bayesian non-singular with respect to $\Pri$ under the GLM.
\item 
If $E_\Pri \{ \log \max_i w_i \} = -\infty$, then $\phi(\xi; \Pri) = -\infty$, i.e.~$\xi$ is Bayesian singular with respect to $\Pri$ under the GLM.  
\end{enumerate}
\label{lemma:glm-sing}
\end{lemma}
Lemma \ref{lemma:glm-sing} can often be used to identify clear conditions on the prior distribution that lead to singularity (or non-singularity). However, to do so it is necessary to analyse the tail behaviour of the GLM weight function, $w(\eta)$, as $|\eta| \to \infty$ in order to establish whether (i) or (ii) above holds. Thus, the results depend upon which link function is chosen. In the remainder of Section \ref{sec:logistic}, results are given for logistic, probit and Poisson regression.}
\par

{\mysubsection{Logistic regression}\\[0em]
For logistic regression, $y_i \given \beta \sim \text{Bernoulli}(\pi_i)$, where $\pi_i = \Pr(y_i=1 \given \beta) = \mu(x_i ; \beta)$. The link function is the logit, $h(\pi)=\log\{\pi/ (1-\pi)\}$, and
\begin{align}
w(\eta) &= \exp( -|\eta|) \expit(|\eta|)^2 \label{eq:logistic-weight} \\
&  \sim \exp(-|\eta|) \text{ as } |\eta| \to \infty \,. \notag
\end{align}
Above, $
\expit(\eta) = 1 / \{ 1+ e^{-\eta} \}$. Lemma \ref{lemma:glm-sing} is now used to establish   sufficient conditions for the prior distribution to be non-singular for logistic regression. 
}
\begin{theorem}
Suppose that $\Pri$ is such that $\E_\Pri ( |\beta_j| ) < \infty$, for $j=0,\ldots,p-1$. If $\xi$ is non-singular for the linear model with regressors given by $\vec{f}$, that is 
$
\left| F^\T F
 \right|>0
$, 
then $\phi(\xi; \Pri) > -\infty$,  i.e. $\xi$ is also Bayesian non-singular with respect to $\Pri$ for the logistic model.
\label{thm:logistic-positive}
\end{theorem}
Note that there is no requirement for $\Pri$ to have bounded support. In particular, this result provides theoretical reassurance that Bayesian $D$-optimality can be used to select a design  with a  normal or log-normal prior  on the parameters. {There is also no requirement for the parameters to be independent a priori. For example, the result applies to a normal-mixture hierarchical variable selection prior distribution (\citet{chipman1997bayesian}).}

Other important prior distributions  do not satisfy the conditions of Theorem 1; for example that proposed by \citet*{gelman2008weakly} which we denote by $\Pri_G$. These authors  recommended rescaling before fitting the model. For observational studies, each explanatory variable is transformed to have mean zero and a standard deviation of 1/2.
This ensures that the method reflects the widely-held default prior belief  that higher order interactions are likely to make a smaller contribution to the linear predictor. The combination of $\Pri_G$ and this scaling was shown to give improved predictive performance relative to both maximum likelihood and penalized logistic regression.  An analogue of the above method for designed experiments is to combine $\Pri_G$ with a standardization of the design variables to have range $[-1/2,1/2]$. This achieves a similar penalization of higher order interactions.  

It is possible to obtain a partial {inverse result} to Theorem \ref{thm:logistic-positive}. 

\begin{proposition}
Given $j \in \{ 0,\ldots,p-1\}$, suppose that: 
\begin{enumerate}[label=\emph{(\roman*)}]
\itemsep 0pt \parskip 0pt \parsep 0pt
\item $\Pri$ is such that $\Pr(\beta_j > 1)>0$
{ \item$\Pri$ is such that, for all $\delta> 0$, 
$$
\Pr(|\beta_k| < \delta \text{ for all } k\neq j \,|\, \beta_j > 1)>0 \,
$$
\item $\Pri$ is such that, for all $\delta>0$,   
$$\E_\Pri [\beta_j \given \beta_j >1,  \, |\beta_k| < \delta,\text{ for all } k\neq j] =\infty$$}
\item $\xi$ is such that $\min_{i=1,\ldots,n} | f_j (\vec{x}_i) | > 0$.
\end{enumerate}
Then $\xi$ is Bayesian singular with respect to $\Pri$, i.e. $\phi(\xi ; \Pri) = -\infty$. 
\label{prop:logistic-negative}
\end{proposition}

{
A more intuitive understanding of the reason that the above conditions lead to a singular prior distribution can be obtained by considering locally optimal design.
There, we have that $|M(\xi;\beta)| \approx 0$ if the responses are close to deterministic, i.e.~if  for all design points the success probability $\Pr(y_i = 1 \given \beta )$  is close to either 0 or 1. In that case, there is also a high probability of separation (\citet{albert1984existence}) and thus non-existence of maximum likelihood estimates.
For Bayesian design, a heavy-tailed prior satisfying the conditions of Proposition \ref{prop:logistic-negative} leads to similarly extreme values of the success probability, which is now a random variable owing to dependence on $\beta$.  Specifically, the implied distribution on $\Pr(y_i = 1 \given \beta)$ has high concentration near either 0 or 1, in the following sense:}
{\begin{proposition}
Under the conditions in Proposition \ref{prop:logistic-negative}, there exists an event $\mathcal{E}\subseteq \Theta$, with $\Pr(\mathcal{E})>0$, 
conditional upon which either $\Pr(y_i = 1 \given \beta)$ or $1-\Pr(y_i = 1 \given \beta)$ has prior geometric mean zero, according to whether $  f_j(x_i)<0$ or $f_j(x_i)>0$ respectively.  
\label{prop:extreme-probs}
\end{proposition} 
 The proofs of Propositions \ref{prop:logistic-negative} and \ref{prop:extreme-probs} both rest on the identification of a region, $\mathcal{E}$, of parameter space where the linear predictor $\eta_i$ can be approximated by the contribution, $\beta_j f_j(x_i)$, from the $j$th predictor.}

The Gelman prior distribution, $\Pri_G$, places independent standard Cauchy distributions on
$(1/10)\beta_0, (2/5)\beta_1, \ldots, (2/5)\beta_{p-1}$.
{Thus, the prior distributions for the regression coefficients are heavy-tailed, with undefined prior mean. The parameters are expected a priori to have large magnitude, i.e.~$\E|\beta_k| = \infty$, $k=0,\ldots,p-1$.} For a model with an intercept term, $f_0(x)=1$, and Proposition \ref{prop:logistic-negative} may be applied with $j =0$; {conditions (ii) and (iii) follow since $\beta_0$ is both heavy-tailed and independent of the other parameters, hence:}

\begin{corollary}
For a logistic model with an intercept term, the prior distribution $\cP_G$ is singular.
\end{corollary}
{Often prior independence of parameters is not a reasonable assumption. For example, \citet{chipman1997bayesian} define a hierarchical variable selection prior in which the probability of an interaction term being active is dependent on whether the parent terms are active, thereby satisfying the weak heredity principle. Proposition \ref{prop:logistic-negative} can be used to show that, for logistic regression, a prior with this hierarchical structure is singular if the prior distribution of the intercept parameter is a mixture of two scaled zero-mode Cauchy distributions rather than a mixture of two scaled zero-mean normal distributions. In this case, the intercept is again both heavy-tailed and (typically) independent of the other parameters.}

For logistic models with a single controllable variable, scalar $x \in \mathcal{X}$, $\mathcal{X}=\mathbb{R}$, Bayesian $D$-optimal design has also been  studied for a different parameterization (for example,~\citet*{chaloner-larntz}):
\begin{equation}
h(\pi_i) = \beta_1(x_i-\mu) \,,
\label{eq:chal-larntz-model}
\end{equation}
which can be obtained from \eqref{eq:linear-predictor} via $\beta_0 = -\beta_1 \mu$. When $\beta_1=0$, $\mu$ is not identifiable and $|M_\theta (\xi;\theta)|=0$ for all $\xi \in \Xi$, with $\theta=(\mu,\beta_1)^\T$, $\Xi=\mathcal{X}^n$.
The following result, which is straightforward to prove using Theorem \ref{thm:logistic-positive}, provides sufficient conditions for a prior distribution to be non-singular for this form of the model.

\begin{proposition}
\label{prop:logistic-other-param}
For the $(\mu, \beta_1)$-parameterization in \eqref{eq:chal-larntz-model}, 
if (i) $\E_{\cP} ( |\mu \beta_1| ) < \infty$, (ii) $\E_{\cP} ( |\beta_1| )<\infty$ and (iii) $\E_{\cP} ( \log|\beta_1| )> -\infty$, then any design with two or more support points is Bayesian non-singular with respect to $\cP$. 
Hence (i)--(iii) are sufficient  for $\cP$ to be non-singular. 
In this case, $\xi$ is Bayesian $D$-optimal  for $(\beta_0,\beta_1)$ if and only if it is Bayesian $D$-optimal for $(\mu,\beta_1)$.
\end{proposition}

{\mysubsection{Probit regression}\\[0ex]
For probit regression, $y_i \given \beta \sim \text{Bernoulli}(\pi_i)$, $\pi_i = \mu(x_i; \beta)$, with link $h(\pi) = \Phi^{-1}(\pi)$, where $\Phi$ is the standard normal c.d.f.. Here,
$$
w(\eta) = \frac{\varphi(\eta)^2}{\Phi(\eta) (1-\Phi(\eta))} \,,
$$
where $\varphi(\eta) = \frac{1}{\sqrt{2\pi}} e^{-\eta^2/2}$ is the standard normal p.d.f..
The following asymptotic approximation holds (e.g.~\citet{abramowitz1972handbook}, p.298)
\begin{align*}
1 - \Phi(\eta) 
& \sim \frac{ 1  }{  \eta\sqrt{ 2\pi }  } 
	e^{  -\eta^2/2  } \quad \text{ as } \eta \to \infty \,.
\end{align*}
Also, as $\eta \to \infty$, $\Phi(\eta) \to 1$, and so by symmetry of $w(\eta)$
\begin{equation}
w(\eta) \sim \frac{1}{\sqrt{2\pi}}|\eta| e^{-\eta^2/2}
\quad \text{ as  }|\eta|\to \infty \,.
\label{eq:probit-tail}
\end{equation}

This asymptotic approximation can be used with Lemma \ref{lemma:glm-sing} to obtain analogues of the results for logistic regression, with different conditions on the prior distribution.

\begin{theorem}
If $E_\Pri |\beta_k \beta_l| < \infty$, for $k,l=0,1,\ldots,p-1$, then $\cP$ is non-singular for the probit regression model.
\label{thm:probit-nonsing}
\end{theorem}

\begin{proposition}
Given $j \in \{ 0, \ldots, p-1\}$, suppose that:
\begin{enumerate}[label=\emph{(\roman*)}]
\itemsep 0pt \parskip 0pt \parsep 0pt
\item $\Pri$ is such that 
$\Pr(\beta_j > 1 ) > 0$
and, for all $\delta>0$,
$$\Pr ( |\beta_k| < \delta \text{ for all }k\neq j \given \beta_j > 1) > 0 $$
\item  $\Pri$ is such that, for all $\delta>0$, 
$$
E_\Pri [\, |\beta_j|^2  \,\big|\, \beta_j> 1, |\beta_k| < \delta \text{ for all } k\neq j] = \infty
$$
\item   $\xi$ is such that $\min_i |f_j(x_i)| >0$.

\end{enumerate}
Then, for the probit link the design $\xi$ is Bayesian singular with respect to $\cP$, i.e. $\E_\Pri \log|M(\xi;\bbeta)| = -\infty$.
\label{prop:probit-sing}
\end{proposition}

\begin{corollary}
For a probit model with an intercept term, the prior distribution $\Pri_G$ is singular. 
\end{corollary}
Again, a heavy-tailed prior on the intercept parameter results in the prior being singular for Bayesian $D$-optimality. The intuitive interpretation is similar to that for the logistic model. Note that $\Pri_G$ would remain singular even if it were made somewhat less heavy-tailed, for example by replacing the Cauchy prior on $\beta_0$ with a $t(2)$ prior. In this case, condition (ii) above will still hold because $\beta_0$ has infinite variance.}
\\[0em]
{\mysubsection{Poisson regression}\\[0ex]
Consider the model $y_i \given \beta \sim \text{Poisson}(\lambda_i)$, with $\mu_i = \lambda_i$ and $h(\mu) = \log \mu$. Optimal designs for this model were considered by \citet{russell-poiss} and \citet{mcgree2012robust}. Here, $w(\eta) = \exp(\eta)$ and we have the following results.

\begin{theorem}
For the Poisson regression model with log link, if $E_\Pri |\beta_k| < \infty$, $k=0,\ldots,p-1$, and $|F^\T F|>0$ then $E_\Pri\log|M(\xi; \bbeta) | > -\infty$. Hence, if the first moments for $\beta_k$ are finite then $\Pri$ is non-singular.
\label{thm:poiss-nonsing}
\end{theorem}

\begin{proposition}
Given $j \in \{ 0,\ldots, p-1\}$, suppose that:
 \begin{enumerate}[label=\emph{(\roman*)}]
 \itemsep 0pt \parskip 0pt \parsep 0pt
 \item $\Pri$ is such that $\beta_j$ is supported on $(-\infty, 0)$ 
 \item $\Pri$ is such that $\Pr(\beta_j < -1)>0$ and, for all $\delta>0$, 
$$
\Pr(|\beta_k| < \delta \text{ for all } k\neq j \,|\, \beta_j < -1)>0 
$$
\item $\Pri$ is such that for all $\delta>0$, 
$$E_\Pri[ \beta_j \given \beta_j < -1,\, |\beta_k|<\delta \text{ for all }k\neq j ] = -\infty$$
\item $\Pri$ is such that $E_\Pri|\beta_k|<\infty$, $k\neq j$
 \item $\xi$ is such that $f_j(x_i) > 0$ for $i =1,\ldots,n $.
 \end{enumerate}
Then $E_\Pri \log|M(\xi; \beta)| = -\infty$, i.e.~$\xi$ is singular  for the Poisson model with log link under Bayesian $D$-optimality.
\label{prop:poiss-sing}
\end{proposition}
\begin{corollary}
For a Poisson model with log link containing an intercept, i.e. $f_0(x)=1$, if $\beta_0$ has a negative half-Cauchy prior independently of $\beta_k$, $k=1,\ldots,p-1$, with  $\E_\Pri|\beta_k| < \infty$, then $\Pri$ is singular. 
\end{corollary}
Here, a heavy-tailed negative intercept parameter can result in a singular prior. Intuitively, it is clear that large negative values of $\beta_0$ will lead to experiments where most of the responses are zero, leading to difficulties obtaining precise estimates of $\beta_0$ and the other parameters.}
\mychapter{Numerical methods to overcome ill-conditioning}
\label{sec:illcond}
\mysection{Objective function approximation}

In performing a numerical search for Bayesian $D$-optimal designs it is necessary to approximate the objective function, usually via a weighted sum,
\begin{equation}
\phi(\xi;\Pri) \approx \phi(\xi; \Q) 
= \sum_{l=1}^{N_\Q} 
	v_l 
	\log|M(\xi ; \bbeta^{(l)} ) | \,,
\label{eq:naive-quadrature}
\end{equation} over a weighted sample, 
$$
\Q = 
\left\{  
\begin{array}{ccc}
\bm{\beta}^{(1)} & \ldots & \bm{\beta}^{(N_\Q)} \\
v_1 &\ldots &v_{N_\Q} 
\end{array}
\right\}\,,
$$
of parameter vectors, $\bm{\beta}^{(l)}$, $l=1,\ldots, N_\Q$, with corresponding integration weights $v_l$,  satisfying $\sum_{l=1}^{N_\Q} v_l = 1$. 

The sample $\mathcal{Q}$ may be obtained, for example, by space-filling criteria, 
as used by \citet*{woods06}, Latin hypercube sampling, or a quadrature scheme, such as that applied by \citet*{gotwalt}. 
Quadrature methods, and in particular the Gotwalt method, can often yield highly accurate approximations. 

A problem with approximation \eqref{eq:naive-quadrature}
is that for multi-parameter models numerical evaluation of $\phi(\xi;\Q)$ can fail due to the presence 
of ill-conditioned matrices $M(\xi; \beta^{(l)})$, whose determinant will be estimated numerically as zero. {Note this can occur even for non-singular $\cP$; for singular $\cP$ there is little point in evaluating $\phi(\xi;\Q)$ since $\phi(\xi;\cP)=-\infty$.}
When numerical evaluation of $\phi(\xi;\Q)$ fails for all $\xi \in \Xi$,  we say that $\Q$ is an \emph{ill-conditioned quadrature scheme}.
{In principle, $\Q$ can be ill-conditioned for any prior distribution. However, for the models considered here, such as logistic regression, ill-conditioning of $\Q$ is more likely if the underlying prior distribution is heavy-tailed. 
In that case, there is high probability of large $\beta$, and so also of $M(\xi;\beta)$ being ill-conditioned.
For any prior, even without heavy tails, other circumstances that may lead to ill-conditioning of $\Q$ include: (i) use of a quadrature scheme, such as the Gotwalt method, which oversamples the tails of $\mathcal{P}$; and (ii) use of a large number of quadrature points. In most integration problems, an increased number of quadrature points leads to improved approximation of the integral; paradoxically, in Bayesian $D$-optimal design this  may cause numerical evaluation to fail due to ill-conditioning. }
\\[0em]
\mysection{Objective function bounds for logistic regression}
\label{sec:logistic-numerical}

For some important models, it is possible to obtain bounds that allow approximation of $\phi(\xi; \Q)$ when $\Q$ is ill-conditioned, {as often occurs for heavy-tailed priors. 
These bounds may be applied to enable straightforward selection of Bayesian $D$-efficient designs for such priors (see Section \ref{sec:logistic-use-of-bounds}).} 
Here we focus on the case of logistic regression, {but a similar approach can be used for the compartmental model (using Lemma \ref{lemma:compart-bounds}), and other GLMs.}
From Lemma \ref{lemma:info-bound} and \eqref{eq:logistic-weight}, we see that $\phi(\xi; \bbeta) = \log |M(\xi;\bm{\beta})|$ lies in $[\phi_L(\xi;\bbeta), \phi_U(\xi;\bbeta)]$, where 
\begin{align*}
\phi_L(\xi; \bbeta) &=  \log|F^\T F|  + p \min_{i=1,\ldots,n} \{ -|\eta_i| + 2\log\expit|\eta_i| \}  \\
\phi_U(\xi; \bbeta) &= \log|F^\T F| + 
 p\max_{i=1,\ldots,n} \{ -|\eta_i| + 2\log\expit|\eta_i| \}  \,.
\end{align*}
Let  $\cS$ be the set of   $l \in \{1,\ldots, N_\Q\}$ for which $M(\xi; \bm{\beta}^{(l)})$ is ill-conditioned, then:
\begin{equation}
\phi_L(\xi ; \Q) \leq \phi(\xi; \Q) \leq \phi_U(\xi; \Q) \,,
\label{eq:bound-phi}
\end{equation}
where 
\begin{align*}
\phi_L(\xi; \Q) 
	&= \sum_{l \in \{1,\ldots,N_\Q\}\backslash \cS } 
		v_l \log |M(\xi ; \bbeta^{(l)})| 
	+ \sum_{l \in \cS} v_l  
	\log|F^\T F|  \\
	& \qquad  
		+	\sum_{l \in \cS} v_l	\,
		p\min_{i=1,\ldots,n} 	 
 	 		\{ -| \vec{f}^\T(\vec{x}_i) \bbeta^{(l)} | + 2\log\expit| \vec{f}^\T(\vec{x}_i) \bbeta^{(l)} |  \}  \allowdisplaybreaks	 \\ 
\phi_U(\xi; \Q) 
	&= \sum_{l \in \{1,\ldots,N_\Q\}\backslash \cS } 
		v_l \log |M(\xi ; \bbeta^{(l)})| 
	+ \sum_{l \in \cS} v_l  
	\log|F^\T F|  \\*
	& \qquad  
		+	\sum_{l \in \cS} v_l	\,
		p\max_{i=1,\ldots,n} 
 	 		\{ -| \vec{f}^\T(\vec{x}_i) \bbeta^{(l)} | + 2\log\expit| \vec{f}^\T(\vec{x}_i) \bbeta^{(l)} |  \} \,.
\end{align*}

The bounds $\phi_L(\xi; \Q)$, $\phi_U(\xi; \Q)$ are   much better  conditioned  than $\phi(\xi; \Q)$. The bounds for $\log|M(\xi;\bm{\beta}^{(l)})|$, $l \in \mathcal{S}$, are often wide. However, as the corresponding $v_l$ is often very small, we may nonetheless obtain from \eqref{eq:bound-phi} a relatively narrow interval for $\phi(\xi;\Q)$. Note that \eqref{eq:bound-phi} specifies an interval that contains the approximation $\phi(\xi;\Q)$, and not necessarily the value of   $\phi(\xi; \Pri)$.

In the remainder of Section \ref{sec:illcond}, we use the following example to show how the bounds enable an extension of the set of prior distributions for which Bayesian $D$-efficient designs can be obtained in practice. We begin by illustrating the use of bounds for the objective function.

\begin{example}
\label{ex:objfn_eval}
Potato-packing  experiment (\citet*{woods06}). We use one of the authors' models, defined  by
\begin{align*}
\fx{} &= (1, x_1, x_2, x_3, x_1 x_2, x_1 x_3, x_2 x_3)^\T \\
\bbeta &= (\beta_0, \beta_1, \beta_2, \beta_3, \beta_{12}, \beta_{13}, \beta_{23} )^\T \,,
\end{align*}
where $q=3$,  $\vec{x} = (x_1,x_2,x_3)^\T$.
We adopt a different prior distribution, namely 
$\log \beta_0 \sim N(-1, 2)$, $\beta_1 \sim N(2, 2)$, 
$\beta_2 \sim N(1, 2)$, $\beta_3 \sim N(-1, 2)$, 
and $\beta_{12},\beta_{13},\beta_{23} \sim N(0.5,2)$ independently. 
{Note that the log-normal prior for the intercept parameter is heavy-tailed.} 
However, from Theorem \ref{thm:logistic-positive}, {the above joint} prior distribution is non-singular.

For a double-replicate of the $2^3$ full factorial design, the value of $\phi(\xi; \Pri)$ was approximated using the Gotwalt quadrature scheme, with 5 radial points and 4 random rotations. Direct numerical evaluation of  $\phi(\xi; \Q)$ failed, {since $\mathcal{S}$ was non-empty: it contained 39 parameter vectors.} However, from \eqref{eq:bound-phi}, $\phi(\xi; \Q) \in [-6.85, -6.78]$.
\end{example}
\par

\mysection{Use of bounds in design optimization and assessment}
\label{sec:logistic-use-of-bounds}

The bounds from \eqref{eq:bound-phi} may also be used within an optimization algorithm to help find Bayesian $D$-efficient designs. The \emph{Bayesian $D$-efficiency} of $\xi$ is
$$
\operatorname{Bayes-eff}(\xi ; \Pri)
= \exp \{  [ \phi(\xi ; \Pri) - \phi(\xi^\ast_\cP ; \Pri )]  /p\} \times 100 \% \,,
$$
where $\xi^\ast_\cP \in \argmax_{\xi \in \Xi} \phi(\xi; \Pri)$ is a Bayesian $D$-optimal design. Bayesian $D$-efficiencies near 100\% indicate that $\xi$ achieves a near-optimal trade-off in performance across the support of the prior distribution for $\bbeta$.

When $\Q$ is well-conditioned, the Bayesian $D$-efficiency may be approximated by numerical search for a $\xi^\ast_\Q \in \argmax_{\xi \in \Xi} \phi(\xi; \Q)$ that maximizes the quadrature approximated objective function, and substitution of the design found into
 $$
\operatorname{Bayes-eff}(\xi ; \Q) = \exp \{  [ \phi(\xi ; \Q) - \phi(\xi^\ast_\Q ; \Q)]  /p\} \times 100 \% \,.
$$
However, if $\Q$ is ill-conditioned, {for example if $\cP$ is heavy-tailed,} then this method fails since (i) $\phi(\xi;\Q)$ cannot be evaluated directly, and (ii) $\xi^\ast_\Q$ cannot be found using a numerical search. We may nonetheless use numerical methods to find designs $\xiQL$ and $\xiQU$ that maximize the lower and upper bounds respectively, 
i.e.{} 
$\xiQL \in \argmax_{\xi \in \Xi} \phi_L(\xi; \Q)$ and 
$\xi^\ast_{\Q,U} \in \argmax_{\xi \in \Xi}\phi_U(\xi; \Q)$. Then a lower bound for the Bayesian efficiency of $\xiQL$ can be approximated, via substitution of the designs found into 
\begin{equation}
\operatorname{Bayes-eff}(\xiQL ; \Q) \geq \exp \{  [ \phi_L(\xiQL; \Q) - \phi_U(\xiQU ; \Q)]  /p\} \times 100 \% \,.
\label{eq:bayes-LB}
\end{equation}
To find exact designs that maximize the bounds, we use a continuous co-ordinate exchange algorithm similar to that of \citet*{gotwalt}.

\begin{excont}[continued]
A co-ordinate exchange algorithm was used, with 100 random starts, to search for $\xiQL$, $\xiQU$ among exact designs with $n=16$ runs.  The quadrature scheme $\Q$ was generated using the Gotwalt method, with 3 radial points and one random rotation, yielding a total of 217 support points for $\Q$. The design $\xiQL$, given in Table \ref{tab:ex1-design}, {is very similar to $\xi^\ast_{\mathcal{Q},U}$: to 2 d.p.~the two are identical. For this $\mathcal{Q}$,  the objective function $\phi(\xi;\mathcal{Q})$ cannot be computed exactly due to ill-conditioning. Thus, given an alternative design $\xi'$, e.g. a 16-run combination of $\xi^\ast_{\mathcal{Q},L}$ and $\xi^\ast_{\mathcal{Q},U}$, it is not possible to evaluate whether $\xi'$ has higher Bayesian $D$-efficiency than $\xi^\ast_{\mathcal{Q},L}$.  However, the lower bound on the Bayesian $D$-efficiency is $\operatorname{Bayes-eff}(\xiQL  ; \Q) \gtrsim 99.4 \%$, so any improvement to be gained by using a different design will be very small.}
\end{excont}

\begin{table}[ptb]
\centering
\begin{tabular}{rrrrrrrrr}

 Run & $x_1$ & $x_2$ & $x_3$ && Run& $x_1$ & $x_2$ & $x_3$ \\ 
  \hline
1 & 0.456 & 1.000 & 1.000 & &  9& -1.000 & -1.000 & 1.000 \\ 
2 & -1.000 & -1.000 & -1.000 && 10 & -0.269 & 1.000 & 1.000 \\ 
3& -1.000 & 0.512 & -1.000 & &11 & 1.000 & -1.000 & -1.000 \\ 
 4& -0.137 & -1.000 & -1.000 && 12 & 1.000 & -1.000 & 0.045 \\ 
 5& 1.000 & -1.000 & 1.000 &  &13& -1.000 & -1.000 & -0.124 \\ 
6 & 1.000 & 1.000 & -1.000 & &14 & 0.085 & -1.000 & 1.000 \\ 
7 & 1.000 & -0.038 & 1.000 & & 15& -1.000 & 1.000 & -0.213 \\ 
8  & -1.000 & 1.000 & 1.000 & & 16 & -0.149 & 1.000 & -1.000 \\ 

\end{tabular}\caption{Example \ref{ex:objfn_eval}, Bayesian design, $\xiQL$, that maximizes the lower bound $\phi_L(\xi;\Q)$. \label{tab:ex1-design}}
\end{table}

\begin{figure}[bhpt]
\begin{center}
\includegraphics[width=0.7\linewidth]{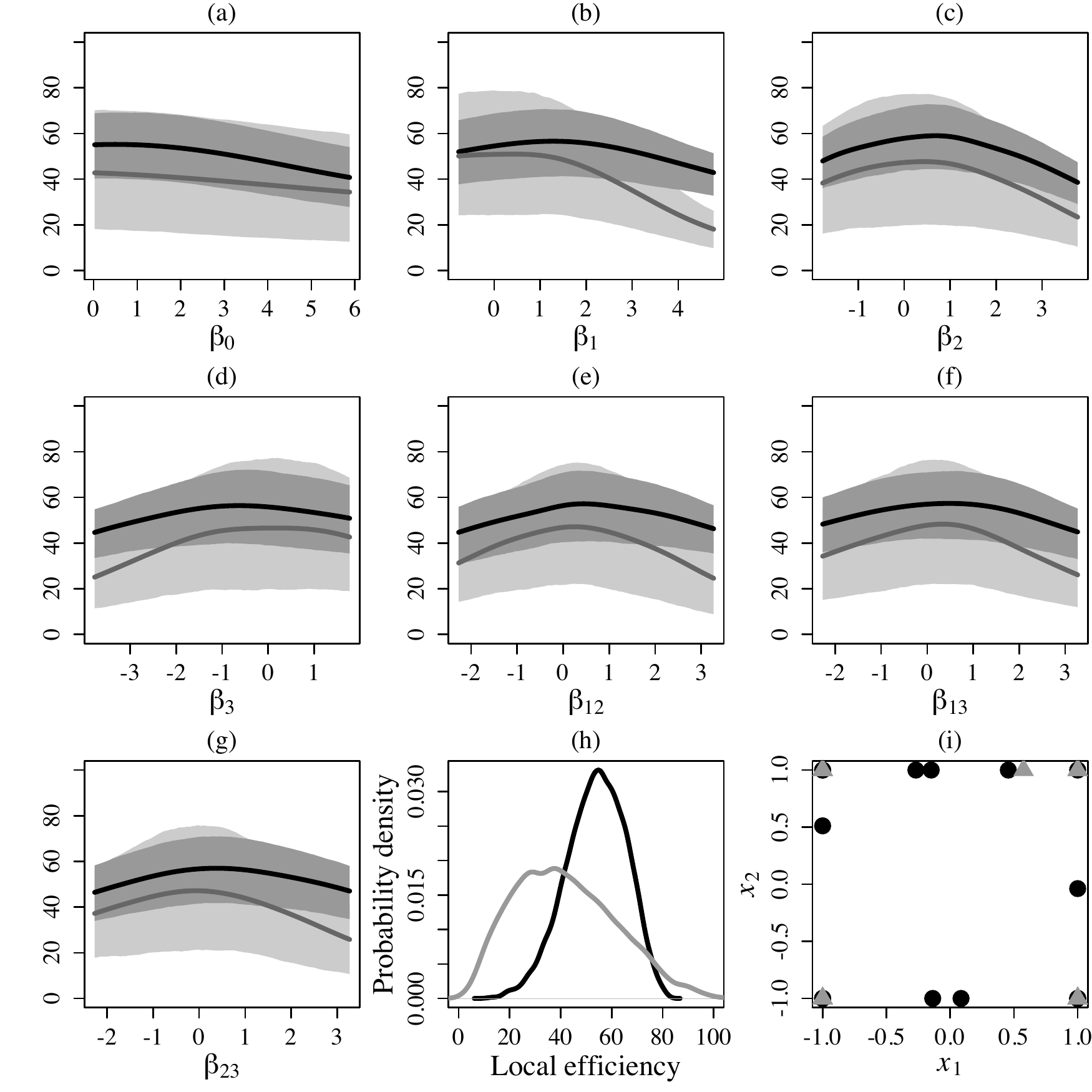}
\caption{ 
{Robustness comparison of the Bayesian $D$-efficient design $\xiQL$ (black lines/points) versus the EW-optimal design (grey lines/points). {\em Panels} (a)--(g): conditional distribution, given $\beta_k$, of the local efficiency, $\eff(\xi ; \beta)$,  induced by the prior on $\beta$  (solid line, conditional mean; shaded region, 10\% and 90\% quantiles). {\em Panel} (h): marginal distribution of the local efficiency. {\em Panel} (i): 2-dimensional projection of the design points.}
\label{fig:ex1-main-effs}}
\end{center}
\end{figure}

Note that the computation of the numerical value of the lower bound in \eqref{eq:bayes-LB} is  approximate since we cannot be certain to have found the global optimum $\xiQU$, although in the above example an assessment of the objective function values from the different random initializations of the algorithm suggests that the number of starts is adequate.

To assess the performance of a given design, $\xi$, for different $\bbeta$, we use the  local $D$-efficiency, 
\begin{equation}
\eff(\xi; \bbeta) =  \{ |M(\xi ; \bbeta)| /  |M(\xi^\ast_\bbeta ; \bbeta)| \}^{1/p}
 \,,
\label{eq:local-eff}
\end{equation}
where $\xi^\ast_{\bbeta} \in \argmax_{\xi \in \Xi} |M(\xi;\bbeta)|$ is a locally $D$-optimal design.
 For some  $\bbeta$, $M(\xi; \bbeta)$ is well-conditioned for most $\xi \in \Xi$. In this case, the local $D$-efficiency can be approximated by searching numerically for $\xi^\ast_{\bbeta}$, and substituting the design found into \eqref{eq:local-eff}.
 For other $\bbeta$, $M(\xi; \beta)$ is  ill-conditioned for all $\xi \in \Xi$. 
 Then  approximate bounds on the efficiency can be derived by numerical search  
 for the designs $\xi^\ast_{U,\bbeta} \in \argmax_{\xi \in \Xi} \phi_U(\xi; \bbeta)$ and 
 $\xi^\ast_{L,\bbeta} \in \argmax_{\xi \in \Xi}\phi_L(\xi;\bbeta)$,  and from the fact that 
\begin{equation}
\exp \frac{1}{p}
[
\phi_L(\xi; \bbeta) - \phi_U(\xi^\ast_{U,\bbeta} ; \bbeta)
]
\leq
\eff(\xi;\bbeta) 
\leq
\exp \frac{1}{p}
[
\phi_U(\xi;\bbeta) - \phi_L(\xi^\ast_{L,\bbeta}; \bbeta)
] \,.
\label{eq:logistic-eff-bound}
\end{equation}

To visualize the dependence of the local efficiency on the individual parameters, {for each regression coefficient $\beta_j$} we plot the approximate {mean and 10\% and 90\% quantiles of the conditional distribution of $\eff(\xi;\beta)$ given $\beta_j$.} Owing to the need to search for a locally $D$-optimal design,  evaluation of $\eff(\xi ; \bbeta)$  is computationally intensive. Thus, before computing {the conditional mean and quantiles}  it is advantageous to first build a statistical emulator  of  $\eff(\xi;\bbeta)$ as a  function  of $\bbeta$, using Gaussian process interpolation. This is analogous to the approach followed in the computer experiments literature when the main effects of a computationally expensive simulator are visualized (e.g.{} \citet*[][Ch.7]{santner2003design}). A similar method was used by \citet*{waite} to visualize the efficiency profile of Bayesian designs for logistic models with random effects.

\begin{excont}[continued]
We consider further the performance of the design, $\xiQL$, that maximizes the lower bound for $\phi(\xi;\mathcal{Q})$. The support points of the quadrature scheme are used to train the emulator of $\eff(\xiQL ;\bbeta)$. In the example, only three out of the 217  $\bbeta$ vectors in $\mathcal{Q}$  led to  $M(\xi; \beta)$ being ill-conditioned for all $\xi \in \Xi$. For these vectors, the efficiency bounds in \eqref{eq:logistic-eff-bound} gave no additional information beyond $\eff(\xiQL;\bbeta) \in [0\%, 100\%]$. Thus we decided to omit these  $\bbeta$ vectors from the training set, as including the bounds $[0\%,100\%]$ would not substantively reduce uncertainty about the efficiency at these $\bbeta$. 
Figure \ref{fig:ex1-main-effs} shows approximations to the conditional mean {and conditional quantiles (given $\beta_j$) of the local efficiency, obtained using the emulator and Monte Carlo sampling.} Also shown is a {kernel density estimate for the marginal} distribution of local efficiencies of $\xiQL$ induced by the prior distribution on $\bbeta$. This is derived by computing the Kriging-based estimates of $\eff(\xiQL ; \bbeta)$ for a Monte Carlo sample of 10,000 $\bbeta$ vectors from the prior distribution. From Figure \ref{fig:ex1-main-effs}, it appears that the modal {local} efficiency of $\xiQL$ is in the range 55-60\%. The lower and upper quartiles of the {local} efficiency distribution are approximately 46\% and 62\%. 
{Although at first glance the typical local efficiencies may appear fairly low, it is important to remember that due to the large amount of prior uncertainty here, there will exist no design whose local efficiency is significantly higher than $\xiQL$  uniformly across the entire parameter space. 
The design $\xiQL$ achieves a near-optimal trade-off in performance, as quantified by the high estimated Bayesian $D$-efficiency obtained earlier,  across the very different parameter scenarios that are possible under the prior for $\beta$. The design is thus relatively robust. There appear to be no significant areas of the parameter space where the design performance is very poor and, for example, the prior probability that $\eff(\xiQL;\beta)< 0.2$ appears negligible.  }

{For comparison, results are included for the EW-optimal design, $\xi^\ast_\text{EW}$, advocated by \citet{yang2016optimal}, which maximizes 
$$
\psi_{\text{EW}}(\xi) 
= \log |E M_\beta(\xi ; \beta) | 
= \log \left| \sum_{i=1}^{n} \E[w(\eta_i)] \,f(x_i) f^\T(x_i) \right| \,.
$$
 In factorial experiments for logistic regression, \citet{yang2016optimal} found EW-optimal designs to be of comparable statistical efficiency to Bayesian $D$-optimal designs, while requiring less computational effort to obtain. Here, as shown in Figure \ref{fig:ex1-main-effs}(a)--(h), the EW-optimal design is much less robust than the Bayesian $D$-efficient design, at least in this case; its local efficiency $\eff(\xi^\ast_\text{EW} ; \theta)$ has generally lower mean, both conditionally on $\beta_j$ and marginally, and the local efficiency also exhibits higher variability. The smaller difference in robustness between Bayesian $D$-optimal and EW-optimal designs observed by \citet{yang2016optimal} may be due to their restriction to a factorial design space: here, the greater performance of the Bayesian $D$-efficient design appears to be due to the inclusion of a greater number of factor settings in the interior of $(-1,1)$ (see Figure \ref{fig:ex1-main-effs}(i)).
 
In addition to the worse statistical performance of the EW-optimal design in this case, here the computational convenience of EW-optimal designs over Bayesian $D$-optimal designs is much reduced. For factorial problems, a reduction in computational cost is achieved by precomputing $\E[w(\eta)]$ for every point in the finite design space, enabling faster evaluation of $\psi_{\text{EW}}(\xi)$. Here, precomputation is not possible since we have continuous factors and  therefore an uncountably infinite design space. One computational benefit of the EW criterion is that it successfully avoids problems with ill-conditioning, but this offers only a minor advantage over Bayesian $D$-optimality, for which ill-conditioning problems can now be overcome using the bounds developed in Section \ref{sec:logistic-numerical}.

 }
\end{excont}
\par

\mychapter{Discussion}
\label{sec:discussion}

{The central tenet of this paper is that it is not permissible to use a singular prior distribution in conjunction with Bayesian $D$-optimality as a design selection criterion. Our new theoretical results, summarized below, can help to ascertain whether a prior is singular or non-singular. This is useful, since if it can be demonstrated that the prior is non-singular, then we may proceed to find Bayesian $D$-optimal designs either using standard methods, or using the new numerical techniques developed in Section \ref{sec:illcond} if problems are encountered with ill-conditioning.  If it is instead demonstrated that the current prior is singular, then $\phi(\xi;\mathcal{P})=-\infty$ for all $\xi \in \Xi$, meaning that design selection fails. One rough intuitive interpretation of this is that the parameter uncertainty under this prior is so great that any design will have low (local) efficiency across a significant portion of the parameter space. In this case, there are two possibilities: (i) consider a different prior distribution, or (ii) adopt a different design selection criterion. These alternatives are considered in Sections \ref{sec:alt-prior} and \ref{sec:alt-criteria} respectively.

To summarize our theoretical results,  for three generalized linear models we have given conditions that can easily be checked to establish non-singularity of $\cP$ and, importantly, identified that a prominent class of default prior distributions for logistic regression should not be used for Bayesian $D$-optimal design. For the compartmental model in Section \ref{sec:compart}, sufficient conditions were established only for singularity of $\cP$, thus highlighting only prior distributions that should not be used. Though desirable, the proof of an inverse result guaranteeing non-singularity seems highly involved and is beyond the scope of this paper.  Future work could seek to develop results on singular prior distributions for population pharmacokinetic models, for which optimal sampling times are more commonly sought (e.g.~\citet{mentre}). Such models extend \eqref{eq:compart-standard} by allowing subject-specific kinetic parameters.}

{\mysection{Alternative prior distributions}
\label{sec:alt-prior}

Often there are multiple plausible candidates for a suitable prior distribution. 
For example, in the subjectivist framework, informative priors are elicited from expert knowledge by obtaining summaries to which a probability distribution may be fitted (e.g.~\citet{garthwaite2005statistical}, \citet{oakley2007uncertainty}). 
Typically there will be multiple distributions that fit the observed summaries. 
Away from this approach, if using uninformative or weakly informative priors there are still often multiple possible candidate priors. 
Thus, if design selection fails because $\mathcal{P}$ is singular but there exists an alternative candidate prior $\mathcal{P}'$ that is non-singular, then Bayesian $D$-optimality may be used with $\mathcal{P'}$ instead. 
Nonetheless, if adopting a subjectivist viewpoint we should be careful to avoid selecting prior distributions purely for analytical convenience if they do not accurately represent the available expert belief or knowledge.

As an example, in Section \ref{sec:exp} we found that $\cP: \theta\sim U(0,a)$ is singular for the exponential regression model.  
A natural question is whether it is sufficient to find designs for the non-singular prior $\cP_\epsilon : \theta \sim U(\epsilon, a)$ for some small value of $\epsilon$ (e.g.~$10^{-3}$ or $10^{-6}$). 
The adequacy of  $\cP_\epsilon$ as a representation of the expert's beliefs will depend substantially on the specifics of the application. 
For small $\epsilon$, the quartiles of $\cP$ and $\cP_\epsilon$ are similar, thus for example it is possible for both distributions to fit expert statements obtained by the bisection method (\citet{garthwaite2005statistical}). 
However, the implication of $\cP_\epsilon$ that there is zero probability that $\theta < \epsilon$ is too strong unless the expert is certain that $\theta\geq \epsilon$. 
The fidelity of the representation $\cP_\epsilon$ would be less important if the resulting design decision were insensitive to the choice of $\epsilon$. 
Unfortunately this is not the case, as shown by the proposition below and its proof (in the supplementary material). 
Intuitively, as $\epsilon \to 0$, some points in the Bayesian $D$-optimal design for $\mathcal{P}_\epsilon$ will converge to zero (while never being equal to zero).
\begin{proposition} 
\label{prop:exp-stability}
For the exponential model, if $\xi$ does not vary with $\epsilon$ then
$$
\operatorname{Bayes-eff}(\xi; \cP_\epsilon) \to 0 \quad \text{as} \quad \epsilon \to 0 \,.
$$
\end{proposition}
Thus, even if one were to compute the Bayesian $D$-optimal design for $\mathcal{P}_{\epsilon'}$, with say $\epsilon'=10^{-6}$, the resulting design would be highly inefficient when evaluated under $\mathcal{P}_\epsilon$  for $\epsilon \ll \epsilon'$.

The situation above is somewhat similar to problems in the objective Bayesian approach with improper uninformative priors (e.g.~\citet[][Ch.3]{berger1985statistical}; \citet{berger2006case}), which one may need to modify in order to obtain a proper posterior. For example, if an improper prior, say $U(10,\infty)$, does not give a proper posterior, one might attempt to replace it with $U(10, M)$, with $M$ large, e.g.~$10^5$ or $10^6$. However, the results would often be highly sensitive to the value chosen for $M$, which is arbitrary and typically has no objective justification.  For further discussion on the role of prior information in design of experiments, see \citet{woods2016bayesian}.

\mysection{Alternative selection criteria}
\label{sec:alt-criteria}

If all candidate prior distributions that agree  with the elicited prior  knowledge or beliefs are singular, then a Bayesian $D$-optimal design cannot be found and it is necessary to use an alternative design selection criterion that suffers  from fewer problems with singularities. 
One such criterion that has already been mentioned  is EW $D$-optimality (\citet{yang2016optimal}). 
Note that the numerical results in Section \ref{sec:logistic-use-of-bounds} suggest that if the problem is with an ill-conditioned $\mathcal{Q}$ rather than a singular $\cP$, then the EW $D$-optimal design may be less robust than a Bayesian $D$-efficient design found using the numerical methods developed here.} Another alternative is to select $\xi$ to maximize the \emph{mean local efficiency},
$$
\Psi(\xi; \Pri) = \E_\Pri \{ \eff(\xi ; \bm{\theta}) \} \,,
\label{eq:mean-local-eff}
$$
which is fairly insensitive to the presence of $\theta$ with $ |M(\xi; \bm{\theta})| \approx 0$. This is a special case of the objective function discussed by \citet*{dette1996optimal} ($\Phi_1$ in their notation). 
{Unlike Bayesian $D$-optimality, neither of the above alternative criteria has} the interpretation of approximate equivalence to the maximization of Shannon information gain. As an example of the use of the mean local efficiency criterion, consider {again} the exponential decay model from Section \ref{sec:exp}.
From Corollary \ref{cor:exp-unif}, when $\mathcal{P}= U(0,a)$, $a>0$, all  designs are Bayesian ($D$-)singular  with respect to $\cP$ for the $\theta$-parameterization. By contrast, it is shown easily that the design with a single support point $x= a/2$ is $\Psi$-optimal with a mean efficiency of approximately 67\%.
This design is locally $D$-optimal when $\theta$ is equal to its prior mean, but highly inefficient when $\theta$ is small. Thus, $\Psi$-optimal designs are much less strongly driven by their worst-case behaviour. {As with EW $D$-optimality, if the problem is in fact with an ill-conditioned $\mathcal{Q}$ rather than a singular $\cP$, then it is possible that $\Psi$-optimal designs may be less robust than Bayesian $D$-optimal or Bayesian $D$-efficient designs.}

A further alternative approach to design selection under parameter uncertainty is to consider maximin designs. In the case of greatest interest in this paper, $\Theta$ is such that  $\inf_{\theta \in \Theta} |M(\xi;\theta)|=0$ for all $\xi \in \Xi$, thus design selection clearly fails using the unstandardized maximin $D$-criterion (\citet{imhof2001maximin}). Often design selection will also fail when using the standardized maximin $D$-criterion.  It is clear that the Bayesian approach, under suitable prior distributions, benefits from greater robustness to the presence of singular $\theta$ than the use of maximin criteria. For related results see \citet{braess2007number}, where conditions are established under which the number of support points in a standardized maximin or Bayesian $D$-optimal approximate design grows arbitrarily large as $\Theta$ is expanded.  In contrast, in the work presented here it is supposed that $\Theta$ is fixed and the focus is instead on examining the adequacy of the set of competing exact designs under different prior distributions. Here, results have also been developed for additional multiparameter models and numerical methods proposed.
\\[1em]

\noindent {\bf Supplementary material}

The online supplementary material for this paper contains proofs of the analytical results described in the text. 
\par
\vspace{1em}

\noindent {\bf Acknowledgement}

I am grateful to {two anyonymous referees whose insightful and helpful comments prompted several improvements to the paper, and to} Professors David Woods and Susan Lewis for several invaluable discussions and comments. This work was supported by the UK Engineering and Physical Sciences Research Council, via a PhD studentship, Doctoral Prize, and a project grant. The work made use of the Iridis computational cluster at the University of Southampton.
\par
\vspace{1em}

\par

\noindent{\bf References}
\small
\bibliographystyle{agsm}
\bibliography{./refs2_SS.bib}

\vskip .65cm
\noindent
 School of Mathematics, University of Manchester, Manchester, M13 9PL, U.K.
 \vskip 2pt
\noindent
E-mail: timothy.waite@manchester.ac.uk

\clearpage
\setcounter{page}{1}
\def\thepage{S\arabic{page}}

\markright{  }

\markboth{\hfill{\footnotesize\rm TIMOTHY W. WAITE} \hfill}{\hfill {\footnotesize\rm SUPPLEMENTARY MATERIALS} \hfill}{}

\renewcommand{\thefootnote}{}
$\ $\par \fontsize{12}{14pt plus.8pt minus .6pt}\selectfont


 \centerline{\large\bf  SINGULAR PRIOR DISTRIBUTIONS  {AND}  }
\vspace{2pt}
\centerline{\large\bf  { ILL-CONDITIONING} IN BAYESIAN \emph{D}-OPTIMAL }
\vspace{2pt}
\centerline{\large\bf  DESIGN FOR {SEVERAL} NONLINEAR MODELS}
\vspace{.4cm}
\centerline{Timothy W. Waite}
\vspace{.4cm}
\centerline{\it  University of Manchester}
\vspace{.55cm}
 \centerline{\bf Supplementary Material}
\vspace{.55cm}
\fontsize{9}{11.5pt plus.8pt minus .6pt}\selectfont
\noindent
\par

\setcounter{chapter}{0}
\setcounter{equation}{0}
\setcounter{lemma}{4}
\def\theequation{S\arabic{chapter}.\arabic{equation}}
\def\thechapter{S\arabic{chapter}}

\fontsize{12}{14pt plus.8pt minus .6pt}\selectfont

\mychapter{Proofs of analytical results}

\begin{proof}[Proof of Proposition \ref{prop:exponential}]
Assume that at least one $x_i >0$.
For the $\theta$ parameterization, we demonstrate two implications: (i) if $\E_\cP(1/\theta)<\infty$ and $\E_\cP(\log\theta)<\infty$, then $\phi(\xi;\cP) > -\infty$; and (ii) if $\E_\cP(\log\theta) = \infty$ or $\E_\cP(1/\theta) = \infty$, then $\phi(\xi;\cP) = -\infty$. Here, $\phi(\xi;\cP) = E \{ \log|M_\theta(\xi;\theta)| \}$, where $\log|M_\theta (\xi;\theta)|$ is given by \eqref{eq:exp-reparam}.

For (i), observe that 
$
-\infty \leq \E_\Pri\left\{ 
(2/\theta)
\max_{i=1,\ldots,n} \{ x_i\}
+ 4 \log \theta
\right\} < \infty$.
Considering the left hand side of \eqref{eq:exp-inequal} and the reparameterization \eqref{eq:exp-reparam}, 
$$
-\infty < 
\log \sum_{i=1}^{n} x^2_i
- \E_\Pri \left\{ 
(2/\theta)
\max_{i=1,\ldots,n} \{ x_i\}
+ 4 \log \theta
\right\} \leq \phi(\xi; \cP) \,,
$$
as required.
For (ii), note that in addition to \eqref{eq:exp-inequal}, the following weaker inequality holds:
$$
\phi(\xi; \cP) \leq
\log \sum_{i=1}^{n} x^2_i
- 4\log \theta \,.
$$
Taking expectations of both sides, if  $\E_{\cP} (\log \theta )= \infty$ then $\phi(\xi;\cP) = -\infty$.

For the other case, let 
$$
b(\theta) =
\frac{1}{\theta} \left\{ 
2 \min_{i=1,\ldots,n} \{x_i : x_i >0 \}
+ 4 \theta \log \theta
\right\} \,.
$$
Since $\theta \log\theta \to 0$ as $\theta \to 0$, there is some $\delta>0$ such that, for $\theta<\delta$, 
\begin{equation*}
b(\theta)  \geq (1/\theta)\min_{i=1,\ldots,n} \{x_i : x_i >0 \} \,.
\end{equation*}
Hence, with $\mathbb{I}$ denoting an indicator function,
\begin{align}
\E_\cP \{ b(\theta) \}& \geq
\E_\cP \{ b(\theta) \mathbb{I}(\theta<\delta) + \inf_{\theta \geq\delta }b(\theta)\, \mathbb{I}(\theta\geq \delta) \} \notag\\
& \geq  \min_{i=1,\ldots,n} \{x_i : x_i >0 \} \E_\cP \{ (1/\theta) \mathbb{I}(\theta < \delta) \} +  (4\log\delta) \Pr(\theta\geq\delta)
\label{eq:exp-prop-inequal1}
\end{align}
If $\E_{\cP} (1/\theta) = \infty$, then  $\E_{\cP} \{ (1/\theta)\mathbb{I}(\theta<\delta)\} = \infty$, and so by \eqref{eq:exp-prop-inequal1}, we have that $\E_{\cP}\{ b(\theta) \}=\infty$, regardless of whether $\E_\cP (\log\theta) = -\infty$.  Recall from \eqref{eq:exp-inequal} that
$$
\phi(\xi; \cP) 
\leq 
\log \sum_{i=1}^{n} x^2_i
-
\E_{\cP}\{ b(\theta) \} \,.
$$
Hence if $\E_{\cP} (1/\theta) = \infty$, we have $\phi(\xi; \cP)=-\infty$. This is sufficient to establish the proposition. 
\end{proof}

\begin{proof}[Proof of Lemma \ref{lemma:compart-bounds}]
Observe that $
M(x_i ; \theta) = e^{-2\theta_1 x_i}
			\tilde{M}^{(i)}_{\delta,\theta_3}
			$, where $\tilde{M}^{(i)}_{\delta,\theta_3}$ is defined in the statement of the lemma.
 Moreover, for $i=1,\ldots,n$, either (i) $x_i=0$ or (ii) $x_i \geq x_{\min}$.  
In (ii), we have
\begin{equation}
e^{-2 \theta_1 x_{\max} } 
		\tilde{M}_{\delta, \theta_3}^{(i)} 
\preceq 
  M(x_i; \theta)
\preceq
	e^{-2 \theta_1 x_{\min} } 
		\tilde{M}_{\delta, \theta_3}^{(i)} \,. 
		\label{eq:lemma1-proof}
\end{equation}
Moreover, the above holds also in (i) since then $M(x_i; \theta)$ and $M^{(i)}_{\delta,\theta_3} $ are matrices of zeroes. Summing \eqref{eq:lemma1-proof} over $i=1,\ldots,n$, we obtain:
\begin{equation}
e^{-2\theta_1 x_{\max} }
\tilde{M}_{\delta, \theta_3}
\preceq
M(\xi; \theta) 
\preceq e^{-2 \theta_1 x_{\min}} 
				\tilde{M}_{\delta, \theta_3} \,.
\label{eq:compart-proof-infobound}
\end{equation}
Taking  log-determinants throughout \eqref{eq:compart-proof-infobound} yields the result, when combined with the fact that $
| \tilde{M}_{\delta,\theta_3} | = \theta_3^4 | \tilde{M}_{\delta, 1}|
$. 
\end{proof}
Define $g_\xi(\delta) = |\tilde{M}_{\delta,1}|$. The following is needed to establish Lemma \ref{lemma:compart-mtilde}.
\begin{lemma}
\label{lemma:compart-derivs}
Suppose that $\xi$ contains at least three distinct $x_i>0$. Then the derivatives of $g_\xi (\delta)$ satisfy: (i) $g_\xi^{(k)}(0)=0$, $k=1,\ldots, 7$, (ii) $g_\xi^{(8)}(0)>0$. 
\end{lemma}
\begin{proof}[Proof of Lemma \ref{lemma:compart-derivs}]
Part (i) can be verified using symbolic computation, e.g. Mathematica. It can also be  shown that 
$$
g_\xi^{(8)}(0)  = 280 \{ S_2 S_4 S_6 
				- S_2 S^2_5
				- S_3^2 S_6
				+ S_3 S_4 S_5
				+ S_3 S_4 S_5
				- S_4^3 \}  \,,
$$
where $S_l = \sum_{i=1}^{n} x_i^l$. Define the following,
$$
K= 
\begin{pmatrix}
S_2 & S_3 & S_4 \\
S_3 & S_4 & S_5 \\
S_4 & S_5 & S_6
\end{pmatrix}\,,
\qquad
K' =
\sum_{i : x_i > 0}
\begin{pmatrix}
	1      	 &   x_i       &    x^2_i \\
	x_i 	 &   x^2_i   &    x^3_i \\
	x^2_i &   x^3_i   &    x^4_i 
\end{pmatrix} \,,
$$
and $x_{\min} = \min \{ x_i : x_i >0 \}$. Note that $
K \succeq x^2_{\min} K' \,.
$
We have
$$
g_\xi^{(8)}(0) = 280 |K| \geq 280 x^6_{\min} |K'| \,.
$$
Observe also that $K'$ is the information matrix of the design $\xi' = (x_i : x_i >0)$ under the linear model with regressors $1, x, x^2$.
By the assumption that there are at least three distinct $x_i>0$, the above linear model is estimable and so $|K'|>0$.  This establishes part (ii).\end{proof}

\begin{proof}[Proof of Lemma \ref{lemma:compart-mtilde}]
If $\xi$ has fewer than three distinct $x_i>0$, then we have $\operatorname{rank}(\tilde{M}_{\delta,1}) \leq 2$ and  $\E_\Pri ( \log |\tilde{M}_{\delta,1}| )=-\infty$ for any prior $\Pri$. Thus we may assume that $\xi$ has at least three distinct $x_i>0$. From Lemma \ref{lemma:compart-derivs}, it is clear that 
$g_\xi (\delta)   \approx   (\kappa/2) 	\delta^8  $  for small $\delta$, 
where $\kappa > 0$. 
We show  that the approximation is sufficiently close that $\E_\cP ( \log|\tilde{M}_{\delta,1} | ) = -\infty$ 
if $\int_{\delta<1} \log \delta \,d\cP(\theta) = -\infty$. 
By Taylor's theorem, there is an $\epsilon_1 >0$ and $\lambda >0$ such that, for $\delta \in (0,\epsilon_1)$,
$$
|g_\xi(\delta) - (\kappa/2) \delta^8 | 
	\leq \lambda \delta^9 \,.
$$
Hence, for $\delta \in (0,\epsilon_1)$, 
$$
 |2g_\xi(\delta)/(\delta^8 \kappa) - 1 | \leq (2\lambda / \kappa) \delta \,.
$$
As  the logarithm function has derivative 1 at argument 1, there exists $0<\epsilon_2 \leq \epsilon_1$ such that for $ \delta \in (0, \epsilon_2)$,
$$
\left| 
 \log \frac{2g_\xi(\delta)}{\delta^8 \kappa }
 - \log 1
\right| 
\leq 2 |2g_\xi(\delta)/(\delta^8 \kappa) - 1 | 
\leq (4\lambda / \kappa) \delta \,.
$$
Thus, for $\delta \in (0,\epsilon_2)$, 
$$
| \log g_\xi(\delta) - \log (\kappa\delta^8/2) | \leq (4\lambda/\kappa) \delta \,,
$$
so that 
$$
\left|
\int_{\delta <\epsilon_2}
 \log g_\xi(\delta)   d\cP(\theta)
 - 
  \int_{\delta < \epsilon_2} 
\{8\log\delta + \log(\kappa /2) \} d\cP(\theta)
 \right| \leq  (2\lambda/\kappa) \epsilon_2^2 \,.
$$
Hence it is clear that $\int_{\delta < \epsilon_2} \log g_\xi(\delta) d\cP(\theta) = -\infty$ if and only if $\int_{\delta < \epsilon_2} \log \delta\, d\cP(\theta) =-\infty$. Further, $g_\xi(\delta)$ is bounded above, and  
$$\int \log g_\xi(\delta) d\cP(\theta)
= \int_{\delta < \epsilon_2} 
	\log g_\xi(\delta) d\cP(\theta) 
	+ \int_{\delta > \epsilon_2} 
	\log g_\xi(\delta) d\cP(\theta) \,.$$ 
Thus, $\int \log g_\xi(\delta) d\cP(\theta) = -\infty$
	when 
	$\int_{\delta< \epsilon_2} \log \delta \,d\cP(\theta)=-\infty$. The result is finally established by  observing that $\int_{\delta< \epsilon_2} \log \delta \,d\cP(\theta)=-\infty$ if we have $\int_{\delta< 1 } \log \delta \,d\cP(\theta)=-\infty$.

\end{proof}

\begin{proof}[Proof of Proposition \ref{thm:compart-model}]
From Lemma \ref{lemma:compart-bounds},
\begin{equation}
\log|M(\xi;\theta)| \leq -6\theta_1 x_{\min}
+ 4\log \theta_3 + \log |\tilde{M}_{\delta,1}| \,.
\label{eq:appx-bound-prop}
\end{equation}
It can be shown that $|\tilde{M}_{\delta,1}| \leq 2S_0 S_2^2 + 4 S_2 S_1^2$, thus $\int \log |\tilde{M}_{\delta,1}| \,d\Pri(\theta)<\infty$. As $\theta_1>0$, $\int-6 \theta_1 x_{\min} d\Pri(\theta) \leq0 <\infty$. If  $\int_{\theta_3> 1} \log \theta_3 \,d\cP (\theta)<\infty$, as assumed by the lemma, then all terms on the right hand side of \eqref{eq:appx-bound-prop} have integral $<\infty$ and
\begin{align*}
\int \log|M(\xi;\theta)| \,d\Pri(\theta) &\leq \int-6 \theta_1 x_{\min} \,d\Pri(\theta)
+ 4\int \log \theta_3\, d\Pri(\theta) \\&\qquad+ \int \log |\tilde{M}_{\delta,1}| \,d\Pri(\theta) \,.
\end{align*}
Hence if, in addition to $\int_{\theta_3> 1} \log \theta_3 \,d\cP (\theta)<\infty$,  we have that at least one of $\int \log |\tilde{M}_{\delta,1}| \,d\Pri(\theta)=-\infty$,  $\int-6 \theta_1 x_{\min} \,d\Pri(\theta)  =-\infty$, or  $\int_{\theta_3 < 1} \log \theta_3 \,d\cP (\theta)=-\infty$ holds, then also $\int \log|M(\xi;\theta)| \,d\Pri(\theta)=-\infty$. Using Lemma \ref{lemma:compart-mtilde}, the condition $\int \log |\tilde{M}_{\delta,1}| \,d\Pri(\theta)=-\infty$ in the preceding statement may be replaced by $\int_{\delta<1} \log \delta \,d\Pri(\theta)=-\infty$. This establishes the result.
\end{proof}

\begin{proof}[Proof of Theorem \ref{thm:logistic-positive}]
{It follows from Lemma \ref{lemma:info-bound} that}
$$
\log|M(\xi ;\bbeta)| \geq \log|F^\T F| + p \min_{i} \log w_i \,.
$$
From \eqref{eq:logistic-weight}, $w(\eta) \geq (1/4) e^{-|\eta|}$. Thus,
\begin{align*}
\log |M(\xi;\bbeta)|
&\geq 
	\log|F^\T F| + p \log\left[(1/4) e^{-\max_i|\eta_i|}\right]  \\
& \geq 
	\log |F^\T F| - p\max_i|\eta_i| - p\log4 \,.
\end{align*}
Moreover, by the triangle inequality, $\max_i |\eta_i| \leq \sum_j \max_i |f_j(x_i)| |\beta_j|$, and hence
\begin{equation}
\log |M(\xi;\bbeta)| \geq \log |F^\T F| - p\log 4 - p \sum_j \max_i |f_j(x_i)| |\beta_j| \,.
\label{eq:appx-lb-logdet}
\end{equation}
The right hand side of \eqref{eq:appx-lb-logdet} has    expectation greater than $-\infty$ due to the assumptions that $\E_\Pri ( |\beta_j| )<\infty$ and $|F^\T F|>0$. Therefore we have that $\E_\Pri \{ \log |M(\xi;\bbeta)| \} > -\infty$. 
\end{proof}

\begin{proof}[Proof of Proposition \ref{prop:logistic-negative}]
From Lemma \ref{lemma:info-bound},  
\begin{align*}
\log |M(\xi;\bbeta)| &\leq \log|F^\T F|   + p \max_i \log w_i \,.
 \end{align*}
 It can also be shown that $w(\eta)$ is a decreasing function of $|\eta|$ and, from \eqref{eq:logistic-weight}, that $w(|\eta|) \leq \exp ( -|\eta|)$. Hence,
 \begin{align*}
\log |M(\xi;\bbeta)| &\leq \log|F^\T F|   + p \log w(\min_{i}|\eta_i|)\\
& \leq  \log|F^\T F|   - p\min_{i}|\eta_i| \,.
 \end{align*}
It remains to prove $\E_\Pri ( \min_{i} |\eta_i| ) = \infty$ to establish that $\E_\Pri\{ \log|M(\xi;\bbeta)| \}= -\infty$. This is achieved by
  conditioning on an event where the parameter $\beta_j$ dominates. Given $j \in \{0,\ldots, p-1\}$, let $\mathcal{E} \in \Sigma$ be an event such that (a) $\beta_j > 1$, and (b) $\sum_{k\neq j} |f_k(x_i)| |\beta_k| < \epsilon$ for all $i$, where $\epsilon>0$ is such that 
$$
| |f_j(x_i)| - |f_j(x_{i'})| | > 2\epsilon \qquad \text{for any }i, i'\text{ with }|f_j(x_i)| \neq |f_j(x_{i'})| \,.
$$
We can guarantee (a) and (b), for example by taking 
\begin{equation}
\mathcal{E} = \{ \bbeta : \beta_j > 1, \,|\beta_k| < \delta\,, \text{ for }k \neq j\}  \in \Sigma\,,
\label{eq:def-E}
\end{equation}
{with $\delta = \epsilon/[(p-1) \max_{i,l} |f_l(x_i)|]$. The above satisfies  
$$
\Pr(\mathcal{E}) = \Pr( \beta_j >1  ) \Pr ( |\beta_k| <  \delta \text{ for all }k\neq j \,|\, \beta_j > 1) >0 \,,
$$ by assumptions (i) and (ii) of the proposition.}

By {the reverse triangle inequality} and from (b), on event $\mathcal{E}$, 
\begin{equation}
| |\eta_i| - |f_j(x_i)| \beta_j |  \leq \sum_{k\neq j} |f_k(x_i)||\beta_k| 
\leq \epsilon \,.
\label{eq:E_inequal}
\end{equation}

Since on $\mathcal{E}$ the term from $\beta_j$ dominates, the minimum of $|\eta_i|$ is found by minimizing the $\beta_j$ term. To see this formally, observe
that if $|f_j(x_i)| \beta_j > |f_j(x_{i'})| \beta_j$, then by the definition of $\epsilon$,
$$
|f_j(x_i)|\beta_j - |f_j(x_{i'})|\beta_j > 2\epsilon \beta_j > 2\epsilon \,.
$$
and, by also using \eqref{eq:E_inequal},
$$
|\eta_{i'}| < |f_j (x_{i'})| \beta_j +\epsilon < |f_j(x_{i})| \beta_j - \epsilon < |\eta_i| \,.
$$ 
{Thus, on $\mathcal{E}$, if $|f_j(x_i)| \beta_j > |f_j(x_{i'})| \beta_j$ then $|\eta_{i}| > |\eta_{i'}|$. 
This can be used to show that on $\mathcal{E}$, if $i^\ast \in \operatorname{arg\,min}_i |\eta_i |$ then $i^\ast \in \operatorname{arg\,min}_i |f_j(x_i)|$, as follows. Suppose that $i^\ast \not\in \operatorname{arg\,min}_i |f_j(x_i)| $, then there would exist some $i$ such that $|f_j(x_{i^\ast})| \beta_j > |f_j(x_{i})| \beta_j$. By the above, we would have that $|\eta_{i^\ast}| > |\eta_{i}|$, which contradicts the definition of $i^\ast$ as a member of $\operatorname{arg\,min}_i |\eta_i |$. Hence,}
\begin{align*}
\min_i |\eta_i| &= |\eta_{i^\ast}|\,, \quad i^\ast \in \operatorname{arg\,min}_i |f_j(x_i)| \\
&\geq |f_j(x_{i^\ast}) | \beta_j - \epsilon \,.
\end{align*}
Consequently,
{\begin{align*}
\E_\Pri ( \min_i |\eta_i| \, \given \, \mathcal{E} ) 
&\geq |f_j(x_{i^\ast})| \, 
		\E_\Pri ( \beta_j \given \mathcal{E} ) - \epsilon  \\
& = \infty \quad \text{ by assumptions (iii) and (iv) of the proposition. }
\end{align*}}
For the marginal expectation, note that $\Pr(\mathcal{E})>0$, and hence
$$
\E_\Pri ( \min_i |\eta_i| ) \geq \Pr(\mathcal{E})\E_\Pri ( \min_i |\eta_i| \given  \mathcal{E} ) = \infty \,.
$$
\end{proof}

{
\begin{proof}[Proof of Proposition \ref{prop:extreme-probs}]
Case (i): assume that $f_j(x_i)> 0$. On the event $\mathcal{E}$ defined in the previous proof, we have that $\eta_i \geq f_j(x_i) \beta_j - \epsilon$. Hence $\E(\eta_i \given \mathcal{E})= \infty$. Let $P= \Pr(y_i=1\given\beta)$, noting that $1 - \Pr(y_i=1 \given \beta) \leq e^{-\eta_i} $. Then,
\begin{align*}
\E^\mathcal{G}( 1- P \given \mathcal{E} ) = \exp \E \log \{ 1- P \given \mathcal{E} \} \leq \exp \E( -\eta_i \given \mathcal{E} ) = 0 \,,
\end{align*}
where $E^\mathcal{G}$ denotes the geometric mean. Hence the conditional geometric mean of $1-P$ is zero.

Case (ii): assume $f_j(x_i)<0$. On $\mathcal{E}$,  $\eta_i \leq f_j(x_i) \beta_j + \epsilon$ and so $\E(\eta_i \given \mathcal{E}) = -\infty$. However, $P \leq e^{\eta_i}$ and so $E^\mathcal{G}( P\given \mathcal{E}) \leq \exp E (\eta_i \given \mathcal{E}) =0$.
\end{proof}

\begin{proof}[Proof of Theorem  \ref{thm:probit-nonsing}]
Assume $\xi$ is such that $|F^\T F| >0$. Note that $w(\eta)$ is decreasing in $|\eta|$ and so, by Lemma \ref{lemma:info-bound},
\begin{align*}
\log |M(\xi;\beta) | &\geq \log|F^\T F| + p \log w(\max_i | \eta_i | )\,.
\end{align*}
We split $\E \log |M(\xi; \beta)|$ into two components, 
\begin{align}
\E \log |M(\xi;\beta) | &= \E \left[ \, \log|M(\xi ;\beta)| \, \mathbb{I}(\max|\eta_i| \leq \kappa) \, \right] \notag \\
& \qquad + \E \left[ \, \log|M(\xi ;\beta)| \, \mathbb{I}(\max|\eta_i| > \kappa) \, \right] \,, \label{eq:split}
\end{align}
where $\kappa > 0$, and then show that both components are $>-\infty$.

 Note that for $|\eta| \leq \kappa$, $w(\eta)$ is bounded below by a constant, $\lambda>0$.
Thus, if $\max_i|\eta_i| \leq \kappa$, then $\log |M(\xi;\beta)| \geq \log |F^\T F| + p\log \lambda$ and so
\begin{equation}
\E \left[ \, \log|M(\xi ;\beta)| \,  \mathbb{I}( \max_i |\eta_i| \leq \kappa) \, \right] > -\infty\,.
\label{eq:ifsmall}
\end{equation}

For $|\eta| > \kappa$, with $\kappa$ sufficiently large, by the asymptotic approximation \eqref{eq:probit-tail},
$$
w(\eta) \geq L |\eta| e^{-\eta^2/2}  
\,,
$$
for some $L>0$.   
 Hence if $\max_i |\eta_i| > \kappa$, then 
\begin{align*}
\log |M(\xi;\beta) | &\geq \log|F^\T F| + p\log L+ p\log \max_i |\eta_i| - p \max_i \eta_i^2 / 2 \\
& \geq \log|F^\T F| + p\log L + p\log \kappa - p \max_i \eta_i^2/ 2 \,.
\end{align*}
However, it is straightforward to show that if $E\beta_k^2 < \infty$ and $E|\beta_k \beta_l| < \infty$ for $k,l=0,\ldots,p-1$, then $\E \max_i \eta_i^2 < \infty$. This is sufficient to prove that 
\begin{equation}
\E \left[ \, \log|M(\xi ;\beta)| \,  \mathbb{I}( \max_i |\eta_i| > \kappa ) \, \right] > - \infty\,.
\label{eq:ifbig}
\end{equation}
Combining \eqref{eq:split}, \eqref{eq:ifsmall} and \eqref{eq:ifbig}, we find that overall $E\log|M(\xi;\beta)| > -\infty$, and so $\Pri$ is non-singular.

\end{proof}

\begin{lemma}
Let $X$ be a random variable taking values in $A \subseteq \mathbb{R}$, with $A$ unbounded above, and let $s, t : A \to \bar{\mathbb{R}}$ be measurable extended real-valued functions  that satisfy (i) for all $k\in \mathbb{R}$, $\sup_{\{x \in A \,|\, x \leq k\} } |s(x)| <\infty$  (ii) $t(x)$ is increasing,  and (iii) $r(x) = t(x)/s(x) \to 0$ as $x\to \infty$. Given the above, if $E[s(X)] = \infty$ then $E[s(X) - t(X)] = \infty$.
\label{lemma:st-infty}
\end{lemma}
\begin{proof}
Note that from (iii) there exists $k'>0$ such that when $X>k'$, we have $s(X)-t(X) \geq (1/2)s(X)$. When $X\leq k'$, by (ii) we have $t(X) \leq t(k')$. 
Hence, 
\begin{align*}
\E[s(X) - t(X)]  &= \E \{ [ s(X)-t(X)] \, \mathbb{I}(X\leq k') + [s(X)-t(X)] \,\mathbb{I}(X> k') \} \\
&\geq E\{ [ s(X) -t(k')]\, \mathbb{I}(X\leq k') + (1/2)s(X) \,\mathbb{I}(X> k') \}\,.
\end{align*}
By condition (i), $s(x)$ is bounded on $\{x \leq k'\}$. Therefore the first term inside the expectation above is also bounded, and since $\E\{s(X)\} =\infty$ we must have that $\E \{ s(X) \mathbb{I}(X>k') \} = \infty$. Hence the right hand side of the above inequality has infinite expectation, and so $\E[s(X) - t(X)] = \infty$.

\end{proof}

\begin{proof}[Proof of Proposition \ref{prop:probit-sing}]
Similar to the proof of Theorem \ref{thm:probit-nonsing}, we split the integral $\E \log |M(\xi; \beta)|$ into two components,
\begin{align}
\E \log |M(\xi;\beta) | &= \E \left[ \, \log|M(\xi ;\beta)| \, \mathbb{I}(\min_i|\eta_i| \leq \kappa) \, \right] \notag \\
& \qquad + \E \left[ \, \log|M(\xi ;\beta)| \, \mathbb{I}(\min_i|\eta_i| > \kappa) \, \right] \,, \label{eq:split-sing}
\end{align}
where $\kappa>0$. Note that $w(\eta)$ is symmetric and decreasing in $|\eta|$. Thus, by Lemma \ref{lemma:info-bound}, for $\min_i |\eta_i| \leq \kappa$ we have $\log |M(\xi; \beta)| \leq \log |F^\T F |  +\log w(0)$, so
$$
 \E \left[ \, \log|M(\xi ;\beta)| \, \mathbb{I}(\min_i|\eta_i| \leq \kappa) \, \right] \leq 
 \Pr( \min_i |\eta_i| \leq \kappa) [ \log |F^\T F | + \log w(0) ] < \infty\,,
$$
i.e. the first term in \eqref{eq:split-sing} is always $<\infty$.

We now consider the second term in \eqref{eq:split-sing}. For $\min_i |\eta_i | > \kappa$, provided $\kappa$ is sufficiently large then by \eqref{eq:probit-tail}, 
$$
\max_i w(\eta_i) = w(\min_i |\eta_i| ) \leq L \min_i|\eta_i| \,e^{-\min_i \eta_i^2/2}\,,
$$
for some $L>0$. By Lemma \ref{lemma:info-bound}, for $\min_i |\eta_i| > \kappa$,
\begin{align*}
\log |M(\xi;\beta)| \leq \log|F^\T F| + p \log L+ p \log \min_i |\eta_i| - (p/2) \min_i\eta_i^2 \,.
\end{align*}
Assume that $|F^\T F | > 0$. Let $X_1 = \min_i |\eta_i| $, $A_1=[0,\infty)$, $s_1(X_1)= (p/2)X_1^2 \,\mathbb{I}(X_1>\kappa) $, and $t_1(X_1) = p\log X_1 \,\mathbb{I}(X_1>\kappa)$. We have that
\begin{align}
&\E \left[ \log|M(\xi;\beta)| \,
	\mathbb{I}(\min_i |\eta_i| > \kappa)
\right]  \notag \\
& \quad \leq
\E \left[ 
 t_1(X_1) - s_1(X_1)
+  (\log |F^\T F|
+ p \log L) \, \mathbb{I}(X_1 > \kappa) 
\right] \,.
\label{eq:probit-sing-part2}
\end{align}
We may assume that $\kappa > 1$, in which case $t_1$ is increasing and so $X_1$, $A_1$, $s_1$, $t_1$ satisfy the conditions of Lemma \ref{lemma:st-infty}. Hence, if $\E [s_1(X_1)]=\infty$ then $\E[t_1(X_1) - s_1(X_1)] = -\infty$, in which case, from \eqref{eq:probit-sing-part2}, 
$$E[ \log|M(\xi;\beta)| \, \mathbb{I}(\min_i |\eta_i| > \kappa)]= -\infty \,,$$ and so by \eqref{eq:split-sing} we have that $\E \log |M(\xi; \beta)|  = -\infty$. Hence to prove the proposition it is sufficient to show that $\E[s_1(X_1)] = \infty$; 
we demonstrate that this holds in the next paragraph.

Recall that on the event $\mathcal{E}$, defined in \eqref{eq:def-E}, we have that $\min_i |\eta_i| \geq \min_i |f_j(x_i)||\beta_j| - \epsilon$. Thus, on $\mathcal{E}$,
\begin{align}
\min_i \eta^2_i &\geq
\min_i |f_j ( x_i )|^2 |\beta_j|^2 -2 \epsilon \min_i |f_j(x_i)| |\beta_j| + \epsilon^2 \notag \\
& \geq s_2(X_2) - t_2(X_2) + \epsilon^2 \,,
\label{eq:minetasq}
\end{align}
where above $X_2 = \min_i |f_j(x_i) ||\beta_j | $, $A=[0,\infty)$, with $s_2(X_2)=X_2^2 $ and $t_2(X_2)=2\epsilon X_2$. 
From assumptions (ii) and (iii) of the proposition we have that  $\E[ |\beta_j|^2 \,\big|\, \mathcal{E}] =\infty$ and $\min_i|f_j(x_i)| >0$, thus 
$$\E[ s_2(X_2) \,|\, \mathcal{E} ] = \E [   \min_i |f_j(x_i)|^2|\beta_j|^2 \,|\, \mathcal{E}] =\infty\,.$$ 
Hence, applying Lemma \ref{lemma:st-infty} we see that $\E[s_2(X_2)- t_2(X_2)\,|\,\mathcal{E}] = \infty$ and so, by \eqref{eq:minetasq}, $\E [ \min_i \eta^2_i \given \mathcal{E}] = \infty$. 
To complete the proof we must consider the marginal expectation of $s_1(X_1) = (p/2)X^2_1 \,\mathbb{I}(X_1 > \kappa)$, where $X_1 = \min_i |\eta_i|$. Note that by assumption (i), $\Pr(\mathcal{E})>0$, thus
$$
\E X^2_1 = \E \min_i \eta^2_i \geq \Pr( \mathcal{E} ) \E( \min_i \eta^2_i \given \mathcal{E} ) = \infty\,.
$$
Finally, observe that   
$X^2_1 = X^2_1 \, \mathbb{I}(X_1 \leq \kappa) + X^2_1 \, \mathbb{I} (X_1 > \kappa)$  and 
$$
0 \leq E\{ X^2_1 \, \mathbb{I}(X_1 \leq \kappa) \} \leq \kappa^2\,.
$$ 
Since $\E X^2_1 = \infty$, we therefore have that $E\{ X^2_1 \,\mathbb{I}(X_1 > \kappa) \} = \infty$. 
Hence $\E[s_1(X_1)]=\infty$. As shown in the previous paragraph, this is enough to establish that $\E \log |M(\xi; \beta)| = -\infty$, and the proposition is proved.

\end{proof}

\begin{proof}[Proof of Theorem \ref{thm:poiss-nonsing}]
From Lemma \ref{lemma:info-bound} and the fact that $w(\eta)= \exp(\eta)$,
\begin{align*}
\log |M(\xi;\beta)| &\geq p \min_i \log w_i + \log |F^\T F| \\
&\geq p \min_i \eta_i + \log|F^\T F| \,.
\end{align*}
However, $\min_i \eta_i \geq -\max_i |\eta_i|$ and so
$$
\log |M(\xi;\beta)| \geq - p \max_i |\eta_i| + \log |F^\T F| \,.
$$
 We know from the proof of Theorem \ref{thm:logistic-positive} that $E \max_i |\eta_i| < \infty$ under the conditions given, and so we also have that $E\log|M(\xi;\beta)| > -\infty$ for the Poisson model.
\end{proof}

\begin{proof}[Proof of Proposition \ref{prop:poiss-sing}]
First note from Lemma \ref{lemma:info-bound} that
$$
\log|M(\xi;\beta)| \leq p \max_i \eta_i + \log|F^\T F| \,.
$$
Thus, to establish that $\E \log|M(\xi;\beta)| = -\infty$, it is sufficient to show that $\E \max_i \eta_i = -\infty$. 
Similar to the proof of Proposition \ref{prop:logistic-negative}, the strategy is to find an event where $\eta_i$ is well approximated by $f_j(x_i)\beta_j$. Let $\mathcal{E}_2$ be an  event  such that  $\beta_j <-1$ and  $\sum_{k\neq j} |f_k(x_i)| |\beta_k| < \epsilon$ for all $i$, where $\epsilon >0$ satisfies
$$
| |f_j(x_i)| - |f_j(x_{i'})| | > 2\epsilon \qquad \text{for any }i, i'\text{ with }|f_j(x_i)| \neq |f_j(x_{i'})| \,.
$$
 For example, one possible definition is
$$
\mathcal{E}_2 = \{  \beta : \beta_j < -1\,,\, |\beta_k| < \delta \text{ for all } k \neq j \} \,,
$$
with $\delta = \epsilon/[(p-1)\max_{i,k}|f_{k}(x_i)|] $.
On $\mathcal{E}_2$, by arguments similar to those in the proof of Proposition \ref{prop:logistic-negative}, 
\begin{align}
\max_i \eta_i &\leq \max_i \{ f_j(x_i)\beta_j\}  + \epsilon \notag \\
& \leq \beta_j \min_i f_j(x_i) + \epsilon \,,
\label{eq:poiss-maxi}
\end{align}
where the second line follows since, by assumptions (i) and (v), we have $\max_i \{ f_j (x_i) \beta_j \}= \beta_j \min_i f_j(x_i) $.
By condition (iii), $\E[\beta_j  \given \mathcal{E}_2 ] = -\infty$ and so, from \eqref{eq:poiss-maxi} and condition (v), we have
$
\E [ \max_i \eta_i  \given \mathcal{E}_2 ] = -\infty
$.
Moreover, by condition (ii), $\Pr(\mathcal{E}_2)>0$ and so 
\begin{equation}
\E[ \max_i \eta_i \, \mathbb{I}(\mathcal{E}_2)] 
= \Pr(\mathcal{E}_2) \,
	\E [ \, \max_i \eta_i  \given \mathcal{E}_2  \,] 
= -\infty\,.
\label{eq:poiss-intmaxioverE}
\end{equation}
Note that
\begin{equation}
  \max_i \eta_i =   \max_i \eta_i \, \mathbb{I}( \mathcal{E}_2)  +  \max_i \eta_i \, \mathbb{I}( \mathcal{E}^C_2) \,.
\label{eq:poiss-maxeta}
\end{equation}
By assumptions (i) and (v),  $f_j(x_i) \beta_j$ is negative and so we have
$
\max _i \eta_i \leq \sum_{k\neq j} \max_i |f_k(x_i)| |\beta_k| \,.
$
Thus by assumption (iv), $\E \{ \max_i \eta_i \, \mathbb{I}(\mathcal{E}^C_2) \}<\infty$. Hence, by \eqref{eq:poiss-maxeta}, 
$$
\E\max_i \eta_i = 
\Pr(\mathcal{E}_2) \,
	\E [\, \max_i \eta_i \given \mathcal{E}_2 \,]
+ \Pr(\mathcal{E}_2^C) \,
	\E [\, \max_i \eta_i \given \mathcal{E}^C_2\,] \,,
$$
and, by \eqref{eq:poiss-intmaxioverE}, we have $\E \max_i \eta_i = \infty$ and hence $\E\log|M(\xi;\beta)| =-\infty$.

\end{proof}

\begin{proof}[Proof of Proposition \ref{prop:exp-stability}]
Let $\zeta_\epsilon$ denote a one-run exact design with design point $x_\epsilon= -1/\log\epsilon$, with $x_\epsilon \to 0$ as $\epsilon \to 0$. We show that, compared to $\zeta_\epsilon$, the relative Bayesian $D$-efficiency under $\mathcal{P}_\epsilon$ of a fixed (exact) design $\xi$ tends to zero as $\epsilon \to 0$. It is not claimed that $\zeta_\epsilon$ is Bayesian $D$-optimal. However, the relative Bayesian $D$-efficiency of $\xi$ is an upper bound for the absolute Bayesian $D$-efficiency of $\xi$, and so this argument is sufficient to establish that $\text{Bayes-eff}(\xi; \mathcal{P}_\epsilon) \to 0$ as $\epsilon \to 0$.

First note that under $\Pri_\epsilon$
\begin{equation}
\E (\beta) = 
	\int_{\epsilon}^{a} \frac{1}{\theta} \frac{1}{a-\epsilon} d\theta 
	= \frac{\log a - \log \epsilon}{a -\epsilon} 
	\to \infty \text{ as }\epsilon \to 0\,.
\label{eq:exp-meanbeta}
\end{equation} Observe that, for the $\beta$-parameterization, using \eqref{eq:exp-inequal},
\begin{align*}
& 
\phi(\xi; \mathcal{P}_\epsilon ) - \phi(\zeta_\epsilon; \mathcal{P}_\epsilon) 
\leq \log S_{xx} - 2 \min_{i: x_i > 0} \{ x_i\} \E \beta - 
2\log x_\epsilon  + 2 \E \beta x_\epsilon 
\\
& \leq \log S_{xx} - 2 \left\{
 \min_{i: x_i >0} x_i - x_\epsilon
\right\} \E \beta
 - 2\log x_\epsilon \,. 
 \end{align*}
 Using \eqref{eq:exp-meanbeta} and the definition of $x_\epsilon$, for $\epsilon$ sufficiently small,
 $$
\phi(\xi; \mathcal{P}_\epsilon) - \phi(\zeta_\epsilon; \mathcal{P}_\epsilon)
 \leq \log S_{xx} + 2 K \log \epsilon - 2 \log \left( \frac{-1}{\log \epsilon} \right) \,,
$$
for some $K>0$. Hence, provided  $\epsilon$ is sufficiently small, the relative Bayesian $D$-efficiency satisfies 
\begin{align*}
\exp\{ 
\phi(\xi; \mathcal{P}_\epsilon ) - \phi(\zeta_\epsilon; \mathcal{P}_\epsilon) 
\} \leq S_{xx} (\epsilon^{K} \log\epsilon)^2 \to 0 \text{ as }\epsilon \to 0\,,
\end{align*}
which is sufficient to prove the claim. The limit above can be found using L'Hospital's rule. Using \eqref{eq:exp-reparam}, the same result for the Bayesian $D$-efficiency also holds under the $\theta$-parameterization. 
\end{proof}}

\end{document}